\documentclass[11pt]{article}

\usepackage{hyperref}

\usepackage{amsmath,amssymb,amsthm,mathtools}
\usepackage{tikz}
\usepackage[capitalize,noabbrev]{cleveref}
\usepackage{fullpage}
\usepackage{authblk}
\usepackage{xspace}

\usepackage[absolute]{textpos}

\newtheorem{theorem}{Theorem}[section]
\newtheorem{lemma}[theorem]{Lemma}
\newtheorem{corollary}[theorem]{Corollary}
\newtheorem{proposition}[theorem]{Proposition}

\newtheorem{claim}{Claim}
\newtheorem{definition}[theorem]{Definition}
\newtheorem{observation}[theorem]{Observation}
\newtheorem{question}[theorem]{Question}

\newcommand\abs[1]{\lvert #1\rvert}

\newenvironment{proofof}[1]{\noindent{\bfseries Proof of #1.}}{\hfill\qed\medskip}

\newcommand{\Oh}{\mathcal{O}}

\newcommand{\A}{\mathbb{A}}

\newcommand{\F}{\mathcal{F}}
\newcommand{\G}{\mathbb{G}}

\newcommand{\N}{\mathbb{N}}

\renewcommand{\P}{\mathcal{P}}

\newcommand{\say}[1]{``#1''} %

\usepackage{xcolor}
\definecolor{brightpink}{RGB}{255,0,204} 
\definecolor{bordeaux}{RGB}{100,0,50}

\newcommand\optional[1]{ {}}
\usepackage{soul}

\newcommand{\yes}{\textsc{Yes}}
\newcommand{\no}{\textsc{No}}

\newcommand{\PLS}{\textsf{PLS}}

\newcommand{\prover}{\mathsf{P}}
\newcommand{\verifier}{\mathsf{V}}
\newcommand{\id}{\mathsf{id}}

\newcommand{\tw}{\mathsf{tw}}
\newcommand{\width}{\mathsf{\omega}}

\newcommand{\margin}{\mathsf{margin}}
\newcommand{\comp}{\mathsf{component}}
\newcommand{\adh}{\mathsf{adhesion}}

\newcommand{\str}{\mathsf{str}}
\newcommand{\parent}{\mathsf{p}}

\newcommand{\roots}{\mathsf{roots}}

\newcommand{\dist}{\mathsf{dist}}

\newcommand{\mso}{\textsf{MSO}$_2$}
\newcommand{\msoo}{\textsf{MSO}$_1$}
\newcommand{\arity}{\mathsf{ar}}
\newcommand{\edge}{\mathsf{edge}}
\newcommand{\vtx}{\mathsf{vtx}}
\newcommand{\edg}{\mathsf{edg}}

\newcommand{\inc}{\mathsf{inc}}
\newcommand{\true}{\textsf{true}}
\newcommand{\false}{\textsf{false}}

\begin{document}

\newcommand\relatedversion{}
\renewcommand\relatedversion{\thanks{The full version of the paper can be accessed at}} %

\title{A tight meta-theorem for LOCAL certification of \mso{} properties within bounded treewidth graphs}%

\date{}

\author[1]{Linda Cook\thanks{Supported by the Institute for Basic Science (IBS-R029-C1) and the gravitation programme NETWORKS (NWO grant no.\ 024.002.003) of the Dutch Ministry of Education, Culture and Science (OCW) and a Marie Skłodowska-Curie Action of the European Commission (COFUND grant no.\ 945045).}}
\author[2]{Eun Jung Kim}
\author[3]{Tomáš Masařík\thanks{Supported by the Polish National Science Centre SONATA-17 grant
number 2021/43/D/ST6/03312.}}
\affil[1]{Korteweg-de Vries Institute of Mathematics,
University of Amsterdam, The Netherlands
}
\affil[2]{School of Computing, KAIST, Daejeon, South Korea / CNRS, France.}
\affil[3]{Institute of Informatics, Faculty of Mathematics, Informatics and Mechanics,\newline University of Warsaw, Warsaw, Poland}

\maketitle

\begin{abstract}
	Distributed networks are prone to errors so verifying their output is critical. Hence, we develop LOCAL certification protocols for graph properties in which nodes are given certificates that allow them to check whether their network as a whole satisfies some fixed property while only communicating with their local network. Most known LOCAL certification protocols are specifically tailored to the problem they work on and cannot be translated more generally. Thus we target general protocols that can certify \emph{any} property expressible within a certain logical framework. We consider Monadic Second Order Logic (MSO$_2$), a powerful framework that can express properties such as non-$k$-colorability, Hamiltonicity, and $H$-minor-freeness. Unfortunately, in general, there are MSO$_2$-expressible properties that cannot be certified without huge certificates. For instance, non-3-colorability requires certificates of size $\Omega(n^2/\log n)$ on general $n$-vertex graphs (Göös, Suomela 2016). Hence, we impose additional structural restrictions on the graph. Inspired by their importance in centralized computing and Robertson-Seymour Graph Minor theory, we consider graphs of bounded treewidth. 

We provide a LOCAL certification protocol for certifying \emph{any}  MSO$_2$-expressible property on graphs of bounded treewidth and, consequently, a LOCAL certification protocol for certifying bounded treewidth. That is for each integer $k$ and each MSO$_2$-expressible property $\Pi$ we give a LOCAL Certification protocol to certify that a graph satisfies $\Pi$ \emph{and} has treewidth at most $k$ using certificates of size $\mathcal{O}(\log n)$ (which is asymptotically optimal). Our LOCAL certification protocol requires only one round of distributed communication, hence it is also \emph{proof-labeling scheme.} 

Our result improves upon work by Fraigniaud, Montealegre, Rapaport, and Todinca (Algorithmica 2024),  Bousquet, Feuilloley, Pierron (PODC 2022), and the very recent work of Baterisna and Chang.
 
\end{abstract}

\section{Introduction}
Distributed communication networks are increasingly vital in practical applications and consequently, the demand for distributed graph algorithms has likewise increased. Communication is a bottleneck in distributed systems, inspiring the need for algorithms where individuals only access local information about their network.
In the 1980s, Linial introduced the LOCAL model of computation to capture this notion of locality \cite{linial-OG-LOCAL-paper-1987, linial-1992-LOCAL-OG-journal-paper}.
 Since then the LOCAL model has become a standard model for distributed graph algorithms. 
A LOCAL algorithm is a distributed graph algorithm running in a \emph{constant} synchronous number of rounds regardless of the network size.
\optional{ possibly we can just reference a definition of the LOCAL model.)}
 
Distributed networks are liable to errors and changes so certifying their solutions and network properties is critical. This motivates the need for \say{LOCAL certification protocols} (also called locally checkable proofs). These are protocols in which vertices are given “certificates of a graph property” that allow the vertices to check whether their network has this global property while only communicating within their local network. 

Informally, a LOCAL certification for a graph property $\mathcal{P}$ is a pair of algorithms called Prover and Verifier. Prover is a centralized algorithm which, upon a graph $G \in \mathcal{P}$, outputs a \say{certificate} for every vertex $v \in V(G)$. Verifier is a LOCAL algorithm that runs on each vertex $v$ simultaneously with some certificate on the local network around $v$ as input.
The key property of a LOCAL certification is that if certificates are corrupted, some vertex in the network should be able to detect this.
That is, in a LOCAL certification of a graph property $\mathcal{P}$ if $G$ does not satisfy $\mathcal{P}$ then for any certificate assignment some vertex in $G$ will output \no \ when running the Verifier algorithm on the local instance.
On the other hand, if $G$ satisfies $\mathcal{P}$ then Prover will provide a truthful certificate, so to speak so that every vertex will output \yes \ when running the verification algorithm on the local instance.  (We postpone all formal definitions until \cref{sec:prelim}.)

Motivated by the fact that distributed networks may need to check their global properties often, LOCAL certification protocols are required to be both time and memory efficient.
Time efficiency is encapsulated by the fact that the Verifier is a LOCAL  (i.e., constant round distributed) algorithm.
When a LOCAL certification protocol requires only exchanging certificates between neighbors, it is called a \emph{proof labeling scheme} \cite{PLS}. Note, a proof labeling scheme runs in a \emph{single} round of distributed communication. 
Our LOCAL certification protocols in this paper are proof labeling schemes.
Since time efficiency is hard-wired into the definition of LOCAL certification,  the performance of a LOCAL certification protocol is measured by its memory efficiency.
The \emph{size} of a LOCAL certification is the size of a maximum certificate assigned by Prover to a vertex  (as a function of the number of vertices in the network).  The primary goal in LOCAL certification is to minimize the \emph{size} of a certification. 
A survey by Feuilloley~\cite{feuilloley2021introduction} can be consulted for further background on LOCAL certification and an exposition of its relationship with other areas including \emph{self-stabilization}.

It is well-known that \emph{every} graph property has a trivial LOCAL certification using  $\mathcal{O}(n^2)$ bits, where $n$ is the total number of vertices  \cite{feuilloley2021introduction}. 
Many interesting graph properties can be certified with just $\mathcal{O}(\log n)$ bits, including  planarity  \cite{feuilloley2020planar}, 
or more generally embeddability on \emph{any} fixed surface \cite{feuilloley2023boundedgenus, louis-graphsonsurfaces}, 
(minimum weight) spanning tree \cite{GoosSuomela16-locally-checkable-proofs, PLS, KormanKutten-spanningtrees}, 
maximum matching \cite{GoosSuomela16-locally-checkable-proofs, CENSORHILLEL-approximatePLS},  and bounded treedepth \cite{treedepth}.
On the other hand, many simple graph properties require $\Omega (\log n)$ bits.
For instance, in \cite{GoosSuomela16-locally-checkable-proofs} G\"{o}\"{o}s and Suomola showed that a LOCAL certification protocol for acyclicity, a spanning tree, or for having an odd number of vertices requires at least  $\Omega(\log n)$ bits.
In the same paper, the authors identified $\mathcal{O}(\log n)$ certificates as a particularly important case. $\mathcal{O}(\log n)$ has since become the gold-standard for certificate size in LOCAL certification;
In the literature, LOCAL certifications are called \emph{compact} if they are of size $\mathcal{O}(\log n)$.

Most LOCAL certification protocols are specifically tailored to the problem they work on and cannot be translated more generally \cite{treedepth}.
Hence, it is desirable to target general protocols in which \emph{all} properties expressible within a particular framework can be certified with compact certificates.
In this paper, we consider general LOCAL certification protocols that can be used to certify any property expressible within \say{monadic second-order logic} (\mso). \mso-expressible properties include non-$k$-colorability, 
non-Hamiltonicity, containing an $H$-minor for any fixed graph $H$,  the existence of a subcubic spanning forest, and, of course, the complements of these properties. 
Some \mso-expressible properties are known \emph{not} to allow a compact LOCAL certification.
For example, certifying non-3-colorability already requires certificates of size $\Omega(n^2 / \log n)$ \cite{GoosSuomela16-locally-checkable-proofs}. 
Hence, we need to impose additional structural conditions on the network in order to obtain compact certifications.

Inspired by their importance in algorithms and in graph structure, we consider bounded treewidth graphs in this paper.
Treewidth is a parameter measuring  \say{close} a graph is to be a tree and having bounded treewidth has many practical and theoretical implications. 
(see e.g.,  Epstein's recent introductory article to treewidth  \cite{treewidth-intro}).
In a seminal result in \emph{centralized} computing, Courcelle showed there is a polynomial-time algorithm to test for any property expressible in \mso \ on graphs of bounded treewidth \cite{Courcelle90}. 
In this paper, we obtain a \emph{tight} LOCAL certification analog to this result.
Our main result is a compact proof labeling scheme for having \emph{any} \mso-expressible  property in addition to having bounded treewidth. 
As is standard in distributed computing all graphs considered here are \emph{finite} and \emph{connected}.
\begin{theorem}\label{main}
For every integer $k$ and for every \mso-sentence $\psi$ on graphs, there is a proof labeling scheme $(\prover, \verifier)$ of size $O(\log n)$ 
for the class of graphs which has  treewidth at most $k$ 
and satisfies the sentence $\psi.$
\end{theorem}

As a consequence, we obtain:

\begin{corollary}\label{thm:treewidthPLS}
	For any positive integer $k$, there is a proof labeling scheme of size  $O(\log n)$ for graphs of bounded treewidth.
\end{corollary}

Note that $\mathcal{O}(\log n)$ is optimal in both \cref{main} and \cref{thm:treewidthPLS}, even if we allow any fixed \emph{constant} number of rounds. Connected graphs of treewidth one are trees. So any LOCAL certification that a graph has treewidth at most one is a LOCAL certification of acyclicity, which requires $\Omega(\log n)$ bits \cite{GoosSuomela16-locally-checkable-proofs}.

\subsection{Relationship with prior work}

Our main result improves upon \cite{logn-squared-treewidth-2024, treedepth, pathwidth}. 
The main result of \cite{logn-squared-treewidth-2024} is a proof labeling scheme of size $\mathcal{O}(\log^2 n)$ for certifying whether a graph has bounded treewidth and satisfies any \mso-expressible property.
In \cite{treedepth} Bousquet, Feuilloley, and Pierron provide an analog of \cref{main} for graphs of bounded \say{treedepth}, and in a very recent paper Baterisna and Chang provide an analog of \cref{main} for graphs of bounded \say{pathwidth.}
	Graphs with bounded treewidth may have arbitrarily large treedepth or pathwidth.
	However, treewidth is an upper bound for both treedepth and pathwidth.
	Hence, \cref{main} is strictly stronger than both of these results.

Additionally, related work has been done on certifying logical properties in other regimes.
In \cite{p4-cw}, Fraigniaud, Mazoit, Montealegre, Rapaport, and Todinca give proof labeling scheme of size $\mathcal{O}(\log^2(n))$ for any \msoo-expressible on graphs of bounded \say{cliquewidth}. \msoo \ is a  slightly more restrictive variant of \mso \. Cliquewidth is a more general graph parameter and gives an exponential upper bound on the treewidth.
In \cite{GoosSuomela16-locally-checkable-proofs}, G\"o\"os and Suomela show that any property expressible in First-Order Logic can be certified compactly.
Note, that First Order Logic is a much more restrictive framework than \mso.

\optional{1 sentence mention of relationship with self stabilization }

\paragraph{Certifying minor-closed properties} 
Robertson-Seymour Graph Minor Theory \cite{lovasz-graphminor-survey} and Graph Product Structure Theory \cite{product-structure} use bounded treewidth graphs as the \say{building blocks} of more general minor-closed classes. Hence, our main result is an important step towards finding a compact certification for (graph) minor-closed properties in general. 
Minor closed properties play an important role in algorithm design and in structural graph theory (see surveys by Bienstock and Langston \cite{bienstock-algorithm-minors}, Lov\'asz \cite{lovasz-graphminor-survey}). 
In their paper presenting a compact LOCAL certification for planarity (a minor-closed property) 
Feuilloley, Fraigniaud, Montealegre, Rapaport, R\'{e}mila, and Todinca   \cite{feuilloley2020planar} 
asked the following question, which has since been reiterated in  \cite{louis-graphsonsurfaces, esperet-minor-local-talk,  bousquet-small-h, feuilloley2020planar, feuilloley2023boundedgenus}.

\begin{question}\label{question:minorclosed}
	Is it true that every minor-closed property has a  compact proof labeling scheme?
\end{question} 

Question \ref{question:minorclosed} is only known in certain special cases, such as bounded genus,  \cite{feuilloley2023boundedgenus, louis-graphsonsurfaces}, $H$-minor-free for certain graphs $H$ on at most five vertices \cite{bousquet-small-h}.
Note, $H$-minor freeness is expressible in \mso \cite{CourcelleE2012}. 
Baterisna and Chang's  proof labeling scheme of \mso \ properties for graphs of bounded pathwidth implies that there is a compact certification for $F$-minor-freeness for any forest $F$ by way of the \say{Excluding a Forest} theorem of Robertson and Seymour \cite{graphminors1-excludingaforest}. 
Our main result extends this to all planar graphs:
\begin{corollary}\label{minorcorr}
For any planar graph $H$ there is a compact certification of $H$-minor-freeness. 
\end{corollary}
\cref{main} implies \cref{minorcorr} by way of a powerful result in Graph Minor Theory. 
In \cite{graphminorsV}, Robertson and Seymour prove that a minor-closed family of graphs $\mathcal{F}$ has bounded treewidth if and only if $\mathcal{F}$ does not contain some planar graph \cite{lovasz-graphminor-survey}.

 In what has since become known as the Robertson-Seymour Graph Minor Theorem, Robertson and Seymour \cite{graphminors-wqo} prove that any minor-closed class can be described by a \emph{finite} list of forbidden minors.
 Hence, \cref{main} also implies that: 
 \begin{corollary}
 	For any minor-closed family of graphs $\mathcal{F}$ there is a compact certification for $\mathcal{F}$ whenever $\mathcal{F}$ does not include \emph{every} planar graph.
\end{corollary}

\section{Preliminaries}\label{sec:prelim}

We use standard terminology for graph theory, e.g.,~\cite{Diestel2005}. For further information on monadic second-order logic on graphs, see~\cite{CourcelleE2012}.

\smallskip

\noindent {\bf Treewidth.} A \emph{tree-decomposition} of a graph $G$ is a pair $(T,\chi)$, consisting of a tree $T$ and a mapping 
$\chi:V(T)\rightarrow 2^{V(G)}$, such that the following conditions are satisfied:
\begin{enumerate}
\item for every vertex $v$ of $G$, there is a node $t$ of $T$ with $v\in \chi(t),$
\item for every edge $uv$ of $G$, there is a node $t$ of $T$ such that $\chi(t)$ contains both $u$ and $v,$ and 
\item for every vertex $v$ of $G$, the set of nodes $t$ of $T$ with $v\in \chi(t)$ is connected in $T.$
\end{enumerate} 

The \emph{width} of a tree-decomposition $(T,\chi)$ is defined as the maximum of $\abs{\chi(t)-1}$ taken over all nodes $t$ of $T$, 
and the \emph{treewidth}, denoted as $\tw(G)$, of a graph $G$ is the minimum width of a tree-decomposition of $G$. 
For a tree-decomposition $(T,\chi)$ of a graph $G$, we call each $t\in V(T)$ a \emph{tree node}, or simply a node, to distinguish it from the vertices of the underlying graph of the tree-decomposition. When $T$ has a distinguished node called the \emph{root}, we say that $T$ is rooted and 
$(T,\chi)$ is said to be a \emph{rooted tree-decomposition}. Throughout this paper, any tree decomposition we consider shall be rooted. 
We also use the standard terminology for a rooted tree such as \emph{parent, child, ancestor, descendant}. 
For a node $t$ of $T$, we denote by $T_t$ the subtree of $T$ rooted at $t$, namely the subtree consisting of tree nodes whose path to the root on $T$ traverses $t.$

We use further notations for  a tree-decomposition $(T,\chi)$ of $G$ from~\cite{Diestel2005,BojanczykP2016}. The set $\chi(t)$ is called the \emph{bag} of $t.$
The \emph{adhesion} of $t$ ($\adh(t)$) is the set $\chi(t)\cap \chi(t')$, where $t'$ is the parent of $t$ if one exists, or defined as $\emptyset$ if $t$ is the root of $T$.
The \emph{margin} of $t$ ($\margin(t)$) is the set $\chi(t)\setminus \adh(t)$.
The \emph{component} of $t$ ($\comp(t)$)is the set $\bigcup_{x\in V(T_t)}\chi(x)\setminus \adh(t)$.
Now, a (rooted) tree-decomposition $(T,\chi)$ is a \emph{sane tree-decomposition} if for every node $t$ of $T$, 
\begin{itemize}
\item the margin of $t$ is nonempty,
\item every vertex in the adhesion of $t$ has a neighbor in the component of $t,$ and
\item the subgraph of $G$ induced by the component of $t$ is connected.
\end{itemize}

It is known that every graph admits an optimal tree decomposition that is sane.

\begin{lemma}[{\cite[Lemma 2.8]{BojanczykP2016}}]\label{lem:optimalsane}\label{BojanczykP2016}
For a tree-decomposition $(T,\chi)$ of $G$ of width at most $k$, 
there is a sane tree-decomposition $(T',\chi')$ of $G$ of width at most $k$ 
such that every bag of $(T',\chi')$ 
is a subset of some bag in $(T,\chi)$. 
\end{lemma}

\noindent {\bf Elimination Tree.} 
Graphs of bounded treewidth can be alternatively witnessed via the notion called an \emph{elimination tree}~\cite{Schreiber82,BojanczykP2017} (also called a \emph{Trémaux tree}), which is also routinely used for defining the treedepth of a graph~\cite{NOdM2012}. 
An \emph{elimination tree} $F$ of a graph $G$ is a rooted tree such that $V(F)=V(G)$ and 
for any edge $uv$ of $G$, either $u$ is an ancestor of $v$ or $v$ is an ancestor of $u$ in $T$.

Let $F_v$ be the subtree of $F$ rooted at $v$. For a vertex $v$ of $F,$ define $\str(v)$ to be the set of ancestors of $v$ (including $v$ itself)
which have a neighbor in $V(F_v).$ 
The \emph{width} of an elimination tree $F$ is defined as 
$$\max_{v\in V(T)} \abs{\str(v)}-1.$$ The \emph{width} of an elimination tree $F$ is the maximum width of 
an elimination tree in $F.$ Given an elimination tree $F$ of width $\width$, 
we can define the set $\chi(v)$ as the union of $\{v\}$ and 
all ancestors of $v$ which have a neighbor in $F_v.$ It is routine to verify that $(F,\chi)$ is a tree-decomposition of width at most $\width.$
It is known~\cite{BodlaenderGHK95, BojanczykP2017} that the treewidth of $G$ equals the minimum width of an elimination forest of $G.$

\bigskip

\noindent {\bf Proof Labeling Scheme.}  
We assume 
that each vertex $v$ of an $n$-vertex  graph is given a unique identifier, written as $\id(v)$, from $\{1,\ldots , n\}$ . %

\medskip

A \emph{labeling} on a graph $G$ is a mapping $\varphi: V(G)\rightarrow \{0,1\}^*$.  
We say that $\varphi(v)$ is a \emph{label} of $v$ of length $\abs{\varphi(v)}$. The \emph{size of a labeling $\varphi$} of $G$ 
is the maximum length of $\varphi(v)$ over all $v\in V(G).$ For an $n$-vertex graph $G$, 
we often express the size of a labeling of $G$ as a function of $n$. That is to say, if $\max_{v\in V(G)} \abs{\varphi(v)}= \Oh(t(n))$, 
we say that $\varphi$ is a labeling of size $t(n)$.

Let $\Pi$ be a graph property and we assume  each graph $G$ of $\Pi$ is equipped with 
an (arbitrary) identifier $\id: V(G)\rightarrow \{1,\ldots, \abs{V(G)}\}$. 
A \emph{prover} $\prover$ for $\Pi$ is an algorithm that produces a labeling on a graph of $\Pi$. 
For an input graph $G\in \Pi$ equipped with $\id$, we denote by $\varphi_{\prover, G,\id}$ 
the labeling $\prover(G)$ produced by the prover $\prover$ run on $G$. 
When the prover $\prover$ and the input graph $G$ with $\id$ are clear in the context, we omit the subscript and simply write as $\varphi.$ 
A \emph{verifier} $\verifier$ is an algorithm which, given a number $i\in \mathbb{N}$ and 
a subset of $\mathbb{N}\times \{0,1\}^*$ as input, outputs \yes\ or \no. 
A graph $G$ equipped with $\id$ and a labeling $\varphi$ on $G$ 
induce  a canonical \emph{local} instance to an arbitrary $\verifier$ at each $v$, defined as 
$$(\id(v), \{(\id(w),\varphi(w) )\}_{w\in N[v]})$$
which we call the \emph{local instance of $G$ and $\varphi$ at $v$}. %

A \emph{proof labeling scheme} (\PLS\ in short) for $\Pi$ 
is  a pair $(\prover,\verifier)$ consisting of a \emph{prover} $\prover$ for $\Pi$ 
and a \emph{verifier} $\verifier$,  satisfy the following two properties together: 

\begin{description}
\item {\bf Completeness.} For every $(G,\id)\in \Pi$, the verifier $\verifier$ run on 
the local instance of $G$ and $\varphi_{\prover,G,\id}$ at $v$ 
outputs \yes\ for each $v\in V(G).$

\item {\bf Soundness.} For every $(G,\id)\notin \Pi$ and for every labeling $\varphi$ on $G$, 
there exists a vertex $u\in V(G)$ such that the verifier $\verifier$ on the local 
instance of $G$ and $\varphi$ at $u$ outputs \no. 
\end{description}

A proof labeling scheme $(\prover,\verifier)$ is said to have size $t(n)$ if the resulting 
labeling $\varphi$ has size at most $t(n)$ for each $n$-vertex graph of $\Pi$.

We make some remarks on the terminology. 
In the context of local/distributed certification, we consider an algorithm that is run at each vertex $v$ based on the information at $v$ (in the form of a label $\varphi(v)$ of $v$) as well as the information locally communicated with $v$. On the other hand, some of the local information may refer to 
a faraway vertex $w$ and neither $v$ nor its neighbors do not even {\sl know} whether $w$ exists in reality or not. On this premise, 
the reference to $w$ at $v$ is made by its identifier. So, $v$ will be aware (by its label) of {\sl some (unknown)} vertex  with identifier $i$, and the vertex $w$ is the very vertex whose identifier is $i.$

\bigskip

\noindent {\bf Monadic Second Order Logic, relational structure, \mso-expressible property.}
A \emph{vocabulary} $\Sigma$ is a set of relation names, or \emph{predicates}. Each predicate $R\in \Sigma$ is associated with a number $\arity(R)\in \N$, called the \emph{arity} of $R$. 

Let $\Sigma$ be a vocabulary. We assume that there is an infinite supply of symbols, all distinct from symbols in $\Sigma$, for \emph{individual variables} and for \emph{unary relation variables} (a.k.a.\ \emph{set variables}). We typically use a lowercase to denote an individual variable and an uppercase letter for a set variable.

A \emph{formula in monadic second-order logic}, simply put \textsf{mso}-formula, over a vocabulary $\Sigma$, is a string that can be built recursively using the logical connectives $\neg, \wedge, \vee, \rightarrow$, the \emph{quantifiers} $\forall, \exists$ as well as the predicates in $\Sigma.$
A string in the form $x=y$, $X(x)$ for some set variable $X$ or $R(x_1,\ldots , x_{\arity(R)})$ for some predicate $R\in \Sigma$ is an \emph{atomic formula}. 
For two formulae $\psi$ and $\phi$, the strings $\neg \psi$, $\psi \wedge \phi$, $\psi \vee \phi$ and $\psi\rightarrow \phi$ are formulae. For a formula $\psi$, 
and for an individual variable $x$ and a set variable $X$, 
$\forall x \psi$, $\exists x \psi$, $\forall X \psi$ and $\exists X \psi$ are all formulae. A variable in a formula $\psi$ is \emph{free} if it does not appear next to a quantifier. We often write the free variables of a formula $\psi$ inside a parenthesis, e.g., $\psi(x,y, Z)$, to highlight the free variables of $\psi$. A formula without a free variable is called a \emph{sentence} and a formula without a quantifier is said to be \emph{quantifier-free}. 
A formula is in \emph{prenex normal form} if it is in the form $Q_1x_1 \cdots Q_\ell x_\ell \psi$ such that $Q_i$ is either the universal quantifier $\forall$ or the existential quantifier $\exists$, $x_i$ is an individual or a set variable, and $\psi$ is quantifier-free.

Let  $\Sigma$ be a vocabulary. For a set $U$, called a \emph{universe of discourse} or simply a \emph{universe}, an \emph{interpretation} of a predicate $R\in \Sigma$ in the universe $U$ is a subset of $U^\arity(R)$.
A \emph{relational structure over $\Sigma$}, or $\Sigma$-structure, is a tuple $\A=(U,(R^\A)_{R\in \Sigma})$ consisting of a \emph{universe} $U$, and an interpretation $R^\A\subseteq U^{\arity(R)}$ for each $R\in \Sigma$. 
 
Fix a vocabulary $\Sigma$ and let $\A$ be a $\Sigma$-structure on the universe $U$. 
An \emph{interpretation} of an \textsf{mso}-formula $\psi$ in $\A$ is an assignment $s$ of each variable in $\psi$ to an element of $U$, in the case of an individual variable, or to a subset of $U$ in the case of a set variable. 
An interpretation $s$ of $\psi$ in $\A$ \emph{evaluates}  to \true, written as $\A\models_s \psi$, or \false, which can be determined as follows.
The atomic formulae $X(x)$, $x=y$ and $R(x_1,\ldots , x_{\arity(R)})$ for some $R\in \Sigma$ will be evaluated \true\ if and only if $s(x)\in s(X)$, $s(x)=s(y)$ 
and $(s(x_1),\ldots , s(x_{\arity(R)})) \in R^\A.$ For a formula such as $\neg \psi$ and $\psi_1\wedge \psi_2$, the evaluation is determined by the evaluation of each subformula (such as $\psi$ in the former, and $\psi_i$ in the latter) in the given interpretation and the logical connective as a boolean operator. For $\psi=\exists x\varphi$, the formula $\psi$ is evaluated to \true\ in an interpretation $s$ 
if and only if there is an interpretation $s'$ with 
$\A \models_{s'} \varphi$ such that $s'$ and $s$ agree on all  variables save $x$.
The formulae of the form $\exists X\varphi, \forall x\varphi, \forall X\varphi$ are evaluated similarly by an interpretation. 
A $\Sigma$-structure $\A$ \emph{models}, or \emph{satisfies}, an \textsf{mso}-sentence $\psi$ if there is an interpretation $s$ of $\psi$ in $\A$ which evaluates to $\true$ and we write $\A \models \psi.$

Of special interest is the vocabulary $\{\edge\}$, consisting of a single predicate  $\edge$ of arity 2.
An undirected graph $G=(V,E)$ can be represented as the $\{\edge\}$-structure $\G=(V, \edge^\G)$ over $\{\edge\}$, where the vertex set $V$ of $G$ is the universe and for every $(u,v)\in V\times V,$
$\edge^\G(u,v)$ if and only if $uv$ is an edge of $G.$ 
Another way to represent a graph $G=(V,E)$ as a relational structure is to represent its incidence graph. We shall denote the associated vocabulary  $\Sigma=\{\vtx,\edg,\inc\}$, where $\vtx$ and $\edg$ are unary predicates and $\inc$ is a binary predicate. 
For a graph $G=(V,E)$, we associate the $\Sigma$-structure $\G=(V\cap E, \vtx^\G, \edg^\G,\inc^\G)$ where $\vtx^G=V, \edg^\G=E$ and 
$\inc^G=\{(v,e)\in V\times E\mid \text{$v$ is an endpoint of $e$ in $G$}\}.$ Such a $\Sigma$-structure is a canonical representation of an incidence graph of $G$ as a relational structure. 

We say that a graph property $\F$ can be expressed in \msoo, or \emph{\msoo-expressible} in short, 
if there is an \textsf{mso}-sentence $\psi$ over $\{\edge\}$ such that 
$G\in \F$ if and only if $\G\models \psi$ where $\G$ is a $\{\edge\}$-structure. 
An \mso-expressible property is defined similarly, for which we take the vocabulary $\Sigma=\{\vtx,\edg,\inc\}$ 
in the place of $\{\edge\}.$ A sentence over $\{\edge\}$, respectively over $\Sigma$, is called an \msoo-sentence, respectively an \mso-sentence, on graphs.

In this paper, we are mostly interested in an \msoo\ or \mso-sentence as a graph property. By abusing the notation, we often write $G\models \psi$ for a graph and an \msoo-sentence, respectively \mso-sentence, $\psi$ as a shortcut to state the following: 
$\G\models \psi$ for an $\Sigma$-structure and a sentence over $\Sigma$, where $\Sigma=\{\edge\}$, respectively $\Sigma=\{\vtx,\edg,\inc\}.$

\section{Key statements and illustrative example}\label{sec:sketch}
In this section, we describe the main ideas of our proof of \cref{main}. 
In the next subsection, we prove a special case of \cref{main}, which illustrates our approach in a toy case. 
In \cref{subsec:keystatements} we explain how to extend this approach to cope with the general case. In particular, we present two key technical statements and show how they imply \cref{main}.
We complete the proof of \cref{main} in Sections 4 and 5 by proving these two key technical statements. 

\subsection{Illustrative example}\label{subsec:subgraphcase}

In this subsection, we present a proof labeling scheme of size $O(\log n)$ for graphs for which there is an optimal elimination tree $T$ for $G$ such that $T$ is a subgraph of $G$. To state formally, let 
$\mathcal G_\width$ be the set of all graphs $G$ such that each component of $G$ admits an elimination tree $T$ of width at most $\width$ as a spanning tree of the component. The key idea developed for this illustrative example will be later generalized for designing a proof labeling scheme for graphs of treewidth at most $\width.$

\begin{proposition}
There is a proof labeling scheme of size $O(\log n)$ for $\mathcal G_\width$.
\end{proposition} 

It suffices to design a proof labeling scheme for a connected graph as such \PLS\ can be applied to each connected component of a graph. Henceforth, 
we assume that the given graph $G$ is connected.

\medskip

\noindent {\bf Prover's construction of a labeling on $G\in \mathcal G_\width$.}  
Let $T$ be an optimal elimination tree of $G$ such that 
$T$ is a subgraph of $G$, and let $r$ be the root of $T$. 
On the input  $(G,\id)\in \mathcal G_\width$, the prover $\prover$ produces a labeling $\varphi$ on $G$ 
so that $\varphi(v)$, for each vertex $v\in V(G)$, encodes the following information. We fix $s:=\lceil \log n \rceil.$

\begin{enumerate}
\item $s$ 1's followed by a prescribed delimiter (e.g., 000111000).
\item The binary representation of $\dist_T(v,r)$ in $s$-bits. 
\item The identifier and the distance $\dist_T(x,r)$ of each vertex $x\in \str_T(v)\setminus v,$  each in $s$-bits.
\end{enumerate}
Note that the distance from any vertex to its root in a tree on at most $n$ vertices does not exceed $n-1$, and can be represented in $s$-bits.
We choose  {\sl not} to broadcast the bound $\width$ on the width as a part of the label. Rather, the number $\width$ will be hardwired in the verifier $\verifier$. 
It is easy to see that:
\begin{observation}
	 If $T$ is an elimination tree of width at most $\width$, it holds that $\abs{\varphi(v)}= \lceil \log n \rceil\cdot (2+2\width)+O(1) \in \mathcal{O}(\log n)$
	for each vertex $v$ of $G.$ 
\end{observation}

\noindent Let us present the verifier's algorithm.

\medskip

\noindent {\bf Verifying the legality of the labels of $v$ and its neighbors.} Let $\varphi$ be a labeling on $G.$
We say that a label $\varphi(v)$ on a vertex $v$ is  \emph{legal} if it is a binary string of the format as what would have been provided by a prover, i.e., 
\begin{center}
[\textsc{\bf Legal label}] a number, say $s_v$, of 1's followed by (a prescribed) delimiter, then followed by a binary string whose length is $2s_v\cdot \width'$ for some $\width'\leq \width.$
\end{center}

First, the verifier $\verifier$ at $v$ examines its own label $\varphi(v)$ and decides if it is legal and outputs \no\ in case of an illegal label. 
For a legal label, $\verifier$ at $v$ parses the string $\varphi(v)$ and interprets each fragment of the string as follows.
\begin{itemize}
\item Let $s_v$ be the value of $\lceil \log \abs{V(G)} \rceil$.
\item Let $d_v\in \mathbb{N}$ be the distance from $v$ to the root on a tree $T'$ presumed by the label $\varphi(v)$.
\item Let $A_v$ be the  set $\str_{T'}(v)\setminus v$, as a set of identifiers, presumed by the label $\varphi(v)$.
\item Let $d_{x,v}$ be the distance assigned to each $x\in A_v$. 
\end{itemize}

During the one-round label communication with its neighbors, $v$ sends its identifier $\id(v)$ and label $\varphi(v)$ 
as a concatenation of the following strings.
\begin{itemize}
\item $s_v$ 1's followed by (a prescribed) delimiter,
\item $\id(v)$ in $s_v$-bits, 
\item $d_v$ in $s_v$-bits,
\item the elements in $A_v$ and $d_{x,v}$ for each $x\in A_v$, each in $s_v$-bits.
\end{itemize}

When the verifier $\verifier$ at vertex $v$ receives a string from its neighbor $u$, 
it first verifies $s_v=s_u$; if not, it outputs \no. Then $\verifier$ at $v$ can parse the string from $u$ 
and extract the tuple $(\id(u),d_u,A_u,\{d_{x,u}\}_{x\in A_u})$ for each neighbor $u$ of $v.$

\bigskip

\noindent {\bf Verifier $\verifier$'s run on a local instance at  $v$.} 
Given the tuples $(\id(u),d_u,A_u,\{d_{x,u}\}_{x\in A_u})$ of all the neighbors $u$ of $v$, $\verifier$ at $v$ 
first identifies \emph{parents} and \emph{children}. A neighbor $u$ of $v$ is a \emph{parent}
(respectively a \emph{child}) of $v$ if $d_u=d_v-1$ (respectively, $d_u=d_v+1$) holds. 
It also identifies a neighbor $u$ as a \emph{strict ancestor} if $d_u<d_v.$

Now $\verifier$ at $v$ verifies the following conditions. 
If any of the conditions (a)-(d) is violated, then it outputs \no\ and outputs \yes\ otherwise.

\begin{enumerate}

\item[(a)] If $d_v>0$, then $v$ has a unique parent $\parent(v)$ of $v$.
\item[(b)] For every child $u$ of $v$, we have $A_u\setminus v \subseteq A_v$. 
Moreover, it holds that $d_{x,u}=d_{x,v}$  for every $x\in A_u\setminus v$ and $d_{v,u}=d_v$.
\item[(c)] For every neighbor $z$ of $v$ with $d_z\leq d_v$, it holds that $z\in A_v$ and $d_z=d_{z,v}.$
\item[(d)] The numbers in the set $\{d_v\}\cup \{d_{x,v}\mid x\in A_v\}$ are all distinct.
\end{enumerate}
This gives a full description of the execution of $\verifier$ at each  $v$ of $G.$
In the remainder of this subsection, we show that the presented \PLS\ is correct.

\bigskip

\noindent {\bf Completeness. }
 One can readily see that for $G\in \mathcal G_t$ with an arbitrary $\id$, the prover $\prover$ generates a labeling on $G$ 
such that $\verifier$ at every $v$ accepts. 

\medskip

\noindent {\bf Soundness.} 
We want to show that if $G\notin \mathcal G_t$, $\verifier$ rejects the local instance at some vertex.
Suppose that for some labeling $\varphi$ on a graph $G$, 
the verifier $\verifier$ at each $v\in V(G)$ accepts its local instance. We will show that $G\in \mathcal G_t$ in this case, thus establishing the soundness of our proof labeling scheme. Notice that there is a single number, say $s$, that all vertices of $G$ agree on, i.e., $s=s_v$ for every $v$ of $G.$ 

Let $F$ be the graph on $V(G)$ obtained by adding an edge between $v$ and the parent $\parent(v)$ of $v$.
By condition (a), there is at most one parent $\parent(v)$ for each vertex $v$, and $\parent(v)$ is a neighbor of $v$ by definition. 
Therefore, $F$ is a subgraph of $G.$  By conditions (a) each vertex with $d_v > 0$ has a unique parent $p(v)$ and the distance $d_{p(v)}$ assigned to $p(v)$ is $d_v - 1$. Hence, $F$ is a spanning forest of $G$ such that each tree of $F$ has a unique vertex $r$, called a root of $F$, with $d_r=0.$ For each vertex $v$ of $G$, let $Q_v$ be the (unique) path from $v$ to the root of the tree in $F$ containing $v.$ We consider $Q_v$ oriented from $v$ toward the root. 

\begin{observation}\label{obs:nosameedge}
For every edge $uv $ of $G$, it holds that $d_u\neq d_v$.
\end{observation}
\begin{proof}
If not, the violation of condition (c) or (d) is detected by $\verifier$ at both $u$ and $v$. 
\end{proof}

\begin{observation}\label{obs:chainpersistent}
Let $v$ be a vertex of $G$ with a neighbor $x$ such that $d_x<d_v.$ Then for every vertex $q$ on $Q_v$ with $d_q\geq d_x$, 
the set $A_q$ contains $x$ unless $q=x.$
\end{observation}
\begin{proof}
Let $q_0,q_1,\ldots , q_\ell$ be the sequence of vertices on $Q_v$, where $q_0=v$ and $q_\ell$ is the (unique) vertex on $Q_v$ with $d_{q_\ell}=d_x.$
By condition (c), it holds that $x\in A_{q_0}.$ By inductively observing that the condition (b) is verified by $\verifier$ at each vertex $q_i$, we conclude that $x\in A_{q_i}$ for all $i=0, \dots , \ell$ or $q_\ell=x.$
\end{proof}

\begin{lemma}\label{lem:elimyes}
$F$ is an elimination tree of width at most $\width.$
\end{lemma}
\begin{proof}
To show that $F$ is indeed an elimination tree of $G$, it suffices to show that there is no edge $uv \in E(G)\setminus E(F)$ 
such that neither vertex is an ancestor of the other in $F.$ In particular, this implies that $F$ consists of a single tree. Indeed, 
if $F$ consists of at least two connected components, then there exist two trees of $F$ as connected components with an edge with exactly one endpoint in each component as $G$ is connected. 
Consider an arbitrary edge $uv$ of $G$. 
We may assume that $d_u>d_v$ by~\cref{obs:nosameedge}. Consider the path $Q_u$ and the vertex $q$ on $Q_u$ with $d_q=d_v.$ Such a vertex $q$ exists on $Q_u$ as $\{d_w\mid w\in V(Q_u)\}=\{0,1,\ldots , d_u\}$.
As the local instance at $u$ meets the condition (c), we have $v\in A_u$. This implies that $v\in A_q$ or $v=q$ by~\cref{obs:chainpersistent}. In the former case, the condition (d) is violated at $q$. Therefore, $v$ is a strict ancestor of $u.$ It follows that $F$ is an elimination tree. 

Now we argue that $\str_F(v)\setminus v\subseteq A_v$ holds for every $v$ of $G$. 
We may assume $\str_F(v)\setminus v\neq \emptyset$, and consider an arbitrary vertex  $x\in \str_F(v)\setminus v.$ 
Note that $x$ has at least one neighbor in $F_v$ and let $u$ be such a descendant of $v$ (possibly $u=v$). 
As $u$ is a strict descendant of $x$ and $v$ lies on the path $Q_u$, we have $x\in A_v$ by~\cref{obs:chainpersistent}.

Finally, observe that $F$ has the width at most $\width$ as $\abs{\str_F(v)}-1=\abs{\str_F(v)\setminus v}\leq \abs{A_v}\leq \width$, where the last inequality holds due to legality condition of the label $\varphi(v)$.
This completes the proof.
\end{proof}

\subsection{Locally verifiable treewidth and key statements}\label{subsec:keystatements}

The result in~\cref{subsec:subgraphcase} holds when $G$ admits 
an elimination forest $F$ of $\width$ as a subgraph of $G$.
In general, however, we cannot expect this condition to hold. In fact, even much weaker versions of this condition may not be satisfiable. In \cite{squirrels-tree-minors-2024} it is shown that for \emph{any} function $f$ there exists a graph $G$ which does not admit tree-decomposition $(T, \beta)$ of width at most $f(\tw(G))$ and $T$ is a minor of $G$. For such graph $G$, any elimination $T$ which is also a minor of $G$ has the width larger than $f(\tw(G))$.

In order to circumvent the issue that the edge of an elimination tree $F$ may be absent in $G$, 
we want to \emph{simulate every edge of $F$ by an actual path in $G$}. 
An  \emph{oriented path system} $\P$ over $G$ is a set of oriented paths between vertices of $G$, where each path in $\P$ is obtained 
by orienting the edges of some path of $G$. The \emph{congestion} of a vertex $v$ of $G$ in the oriented path system $\P$ 
is the number of paths in $\P$ which traverses $v$ as a starting or internal vertex. Note that 
even if $v$ has congestion at most $\width$, there may be arbitrarily many paths in $\P$ having $v$ as a destination vertex. 
The \emph{congestion} of $\P$ 
is defined as the maximum congestion taken over all vertices of $G.$  
An oriented path system $\P$ over $G$ is to said to \emph{witness an elimination tree} $T$ if the following holds.
\begin{center}
  $(\star)$\phantomsection\label{def:star} for every vertex $v$ with a parent $\parent(v)$ in $F$, \\
either $v$ and $\parent(v)$ are adjacent in $G$ or there is a path in $\P$ oriented from $v$ to $\parent(v)$. 
\end{center}
We say that a graph $G$ has \emph{locally verifiable treewidth at most $k$} if $G$ admits an elimination tree $F$ of width at most $k$ and 
an oriented path system $\P$ over $G$ of congestion at most $k$ witnessing $F$ for each connected component. 

The proof labeling scheme for $\mathcal G_t$ presented in the previous subsection is not difficult to generalize for graphs of small locally verifiable treewidth.

\begin{proposition}\label{prop:locallyPLS}
For each $\width$, there is a proof labeling scheme $(\prover, \verifier)$ of size $O(\log n)$ for the class of graphs 
of locally verifiable treewidth at most $\width$.
\end{proposition}
\begin{proof}
Although it is standard to assume the graph is connected in distributed, we note that here the disconnected case reduces to the connected case. 
Indeed, the prover applies the labeling on each connected component. If the verifier at some vertex of a (not necessarily connected) graph outputs \no, then the corresponding connected component has locally verifiable treewidth larger than $\width$ and so does the entire graph. If $\verifier$ outputs \yes\ at each vertex, each component has locally verifiable treewidth at most $\width$, and so does the entire graph.  Henceforth, we assume that the underlying graph $G$ is connected.

Clearly, any graph of $\mathcal G_\width$ handled in the previous subsection has bounded locally verifiable treewidth. 
We modify the proof labeling scheme of Subsection~\ref{subsec:subgraphcase} so that it works for the general case. 
To present a prover $\prover$ for connected graphs of locally verifiable treewidth at most $\width$, 
we fix an elimination tree $T$ of $G$ of width at most $\width$ and an oriented path system $\P$ over $G$ which witnesses $T$ and has congestion at most $\width.$ 
We may assume that each path in the path system $\P$ connects two vertices that are 
adjacent in $T$ but not in $G.$ A $(u,v)$-path $P\in \P$ is presumed to be  oriented  from $u$ toward $v$, where $v$ 
is the parent of $u$ in $T$.

\medskip

Intuitively, we want to open a `channel' connecting $u$ and $v$, using the $(u,v)$-path $P$ in $\P$, which would 
convey the information that $u$ and $v$ would have exchanged if they were adjacent in $G.$ 
The channel between $u$ and $v$ is implemented by adding to each vertex on $V(P)\setminus \{v\}$ 
the tuple $(\id(u),d_u, A_u, \{d_{x,u}\}_{x\in A_u})$, the (identifiers of the) \emph{child} $u$ and \emph{parent} $v$ together with the \emph{succeeding} and \emph{preceding} vertex on $P$ if one exists.
Each vertex on $P$, except for the destination vertex $v$, will compare what its predecessor conveys with its own information and outputs \no\ when there is a mismatch. As each vertex participates in at most $\width$ channels, the size of the label under this modification will 
be increased by an additive factor of $O(\width^2\cdot \log n)$. 

\bigskip

Formally, for every $(u,v)$-path $P=p_0,\ldots , p_\ell \in \P$, 
the prover $\prover$ adds to the label $\varphi(p_i)$ the following information for all $0\leq i <\ell,$ on top of what is encoded already in the case of $\mathcal G_\width$ in the previous subsection.

\begin{itemize}
\item \textsc{Channel name}: the identifiers of $u$ and $v$, where the second element 
$v$ is interpreted as the parent of the first element $u$ in the presumed elimination tree.  
\item \textsc{Predecessor and successor}: the identifiers of $p_{i-1}$ and $p_{i+1}.$ When $i=0$, 
the former will be omitted.
\item \textsc{Cargo}: $(\id(u),d_u, A_u, \{d_{x,u}\}_{x\in A_u})$ and $(\id(v),d_v, A_v, \{d_{x,v}\}_{x\in A_v})$ 
\end{itemize}

If $p_i$ carries multiple channels, each channel can be separated with a suitable delimiter, e.g., 0's of length $s.$ 
The verifier $V$ at a vertex $v$ will verify the following for each channel it carries (on top of what the verifier does in the case of $\mathcal G_\width$ in the previous subsection).

\begin{enumerate}
\item If it is a starting vertex $p_0$ of the channel with $p_1$ as the successor indicated in this channel, 
$p_1$ must be a neighbor of it and have the same channel. There must be no other neighbor carrying the same channel. Moreover, the cargo in this channel at $p_0$ 
should match the cargo of the same channel at $p_1.$
\item If it is an internal vertex $p_i$ i.e., the identifier of the current vertex 
does not match the starting and destination vertices indicated in the channel name, 
then $p_i$ must have precisely two neighbors, whose identifiers match those of $p_{i-1}$ and $p_{i+1}$ 
indicated in the channel of $p_i$, carrying the same channel. 
Moreover, the cargo in this channel at $p_i$ 
should match the cargo  of the same channel at $p_{i-1}$ and $p_{i+1}.$
\end{enumerate}

Finally, a vertex $v$ receives all the cargo from each of its neighbors $w$ which carries a channel with $v$ as the destination vertex. 
If there is a channel named $(u,v)$ in $\varphi(w)$, and the second part of the cargo in the channel indeed does not match $(\id(v),d_v, A_v, \{d_{x,v}\}_{x\in A_v})$, then $\verifier$ at $v$ outputs \no.
It is tedious to verify that both the completeness and the soundness conditions are satisfied. 
\end{proof}

~\cref{prop:locallyPLS} combined with the combinatorial statement below asserting that any graph of treewidth at most $\width$ has small locally verifiable treewidth, we already obtain an `approximate proof labeling scheme'. To turn it into an exact proof labeling scheme for treewidth, we use a proof labeling scheme that verifies if a (connected) graph has an \mso-expressible property; in particular, the property of having treewidth at most $\width.$ 
We formalize the key technical statements and defer their proofs to subsequent sections.

\begin{proposition}\label{thm:tw2lvtw}
There is a function $f(\width) \in 2^{2^{O(\width^2)}}$ such that for every $\width$, any graph of treewidth at most $\width$ has locally verifiable treewidth at most $f(\width).$ 
\end{proposition}

\begin{proposition}\label{thm:msoPLS}
For every integer $k$ and for every \mso-sentence $\psi,$ there is a proof labeling scheme $(\prover, \verifier)$ of size $O(\log n)$ 
for the class of connected graphs which has locally verifiable treewidth at most $k$ 
and satisfies the sentence $\psi.$
\end{proposition}

\begin{proofof}{\cref{main}}
The proof labeling scheme $(\prover, \verifier)$ works as follows. For a graph $G$ of treewidth at most $k$, $G$ has locally verifiable treewidth at most $f(k)$ by~\cref{thm:tw2lvtw} for some function $f$. Therefore, the prover $\prover$ can compute the labeling $\varphi_{local, f(k)}$ on $G$ which certifies that $G$ has locally verifiable treewidth at most $f(k)$, where the existence and computability of $\varphi_{local,f(k)}$ due to~\cref{prop:locallyPLS}. 

Fix an integer $k\in \N$ and let $\psi_k$ be the \mso-sentence such that for any graph $G$, the treewidth of $G$ is at most $k$ if and only if $G$ satisfies $\psi_k.$ As the property of having treewidth at most $k$ is closed under taking a minor, there exists a finite set ${\cal H}_k$ of minor obstructions for treewidth at most $k$ by the Robertson-Seymour Graph Minor Theorem \cite{lovasz-graphminor-survey}. That is, any graph $G$ has treewidth at most $k$ if and only $G$ does not contain any graph in ${\cal H}_k$ as a minor. 
For any fixed graph $H$, the property of not having $H$ as a minor is \mso-expressible. Moreover,  ${\cal H}_k$ is computable for each $k$~\cite{Bodlaender96,Lagergren98}. Therefore, there exists an \mso-sentence defining the property of having treewidth at most $k$, and such a sentence is computable. \optional{Find a reference to an algorithm computing this.. i think I came across one recently} 

Then $\prover$ computes the aforementioned sentence $\psi_k$, and computes the labeling $\varphi_{f(k),\psi'}$ on each connected component of $G$ as given in~\cref{thm:msoPLS}. Here, $\psi':=\psi_k\wedge \psi.$
The label $\varphi_{local,f(k)}(v)$ together with $\varphi_{f(k),\psi'_k}(v)$ is given to each vertex $v$ of $G.$ The verifier $\verifier$ simply checks the first part of the labeling as in~\cref{thm:tw2lvtw} and the second part of the labeling as in~\cref{thm:msoPLS}. The completeness and soundness of the constructed $(\prover,\verifier)$ follows immediately from~\cref{prop:locallyPLS} and~\cref{thm:msoPLS}.
\end{proofof}

In the remaining part of this paper, we focus on establishing~\cref{thm:tw2lvtw} and~\cref{thm:msoPLS}, which are presented in~\cref{sec:locallytw} and~\cref{sec:mso}, respectively.

\section{Proof of~\cref{thm:tw2lvtw}.}\label{sec:locallytw}
To prove \cref{thm:tw2lvtw}, we use the notion of a guidance system introduced by Bojańczyk and 
Pilipczuk~\cite{BojanczykP2016}. 
A \emph{guidance system} $\Lambda$ over a graph $G$ is a family of in-trees of $G$, where 
each in-tree of $\Lambda$ is an oriented tree obtained by orienting all the edges of a subtree of $G$ toward a unique distinguished vertex denoted as the \emph{root}. 
For a vertex $x$ of $G$, we define the set $\roots_\Lambda(x)$ as 
$$ \{y\in V(G):\text{$y$ is a root of an in-tree in $\Lambda$ containing $x$}\}.$$
A guidance system $\Lambda$ \emph{captures} a vertex set $A\subseteq V(G)$ if there is a vertex $u$ such that 
$A\subseteq \roots_\Lambda(u).$ We also say that $u$ captures $A$ in the guidance system $\Lambda.$ 
For a family $\mathcal A$ of vertex subsets of $V(G)$, a guidance system $\Lambda$ \emph{captures 
$\mathcal A$} if every set $A\in \mathcal A$ is captured by $\Lambda.$ 

We say that two in-trees of $G$ \emph{intersect} if they share a vertex which is not the root in both in-trees. Notice that if two in-trees share a vertex $v$ which is the root in one in-tree and not the root in another one then they intersect. 
A guidance system $\Lambda$ is \emph{$k$-colorable} if its intersection graph admits a proper $k$-coloring, or equivalently, 
one can assign colors to the in-trees of $\Lambda$ so that two in-trees have the same color only if they are vertex-disjoint 
or the unique common vertex is the root in each in-tree.

\begin{lemma}\label{lem:saneness}
Let $(T,\chi)$ be a sane tree-decomposition of $G$ and 
let $G^\star$ be the graph obtained from $G$ by 
adding an edge between any pair $u,v$ of vertices whenever there is a tree node $t$ with $u,v\in \adh(t).$
The following properties hold for any tree node $t.$
\begin{enumerate}
\item If $t$ is a tree node with a parent $\parent(t)$, then $\adh(t)\setminus \adh(\parent(t))\neq \emptyset.$
\item For any $w_t\in \adh(t)$, there is a vertex $r_t \in \margin(t)\cap N_{G^\star}(w_t)$.
\end{enumerate}
\end{lemma}
\begin{proof}
For the first property, observe that the saneness of $(T,\chi)$ indicates that both $\margin(t)$ and $\margin(\parent(t))$ are non-empty by the first condition of saneness. If $\adh(t)\subseteq \adh(\parent(t))$, then $\adh(\parent(t))$ separates the vertex sets $\margin(t)$ from $\margin(\parent(t))$. As $\margin(t)\cup \margin (\parent(t))$ is contained in the component of $\parent(t)$, this means that $\comp(\parent(t))$ has at least two connected components, violating the third condition of saneness. 

Let us prove the second property. 
If $w_t$ has a neighbor in $\margin(t)$ in the graph $G$, there is nothing to prove. 
So we may assume that this is not the case. 
By the second condition of saneness, every vertex in the adhesion of $t$ has a neighbor in the component of $t.$ Let $s_t\in \comp(t)$ be a neighbor of $w_t$ in $G$. If $s_t$ is in the margin of $t$, we are done. So suppose $s_t\in \comp(t)\setminus \margin(t).$ By definition $\comp(t)$ can be partitioned into $\margin(t)$ and $\comp(t')$ over all children $t'$ of $t$.
Since $s_t$ and $w_t$ are adjacent, it follows that
there exists a child $b$ of $t$ whose component contains $s_t$ and we have $w_t\in \adh(b)$. By the first property of the lemma statement, there exists a vertex $r_t\in \adh(b)\setminus \adh(t)$ and thus $r_t\in \margin(t)$ . 
It remains to observe that $r_t$ and $w_t$ are adjacent in $G^\star$ since $r_t, w_t \in \adh(b)$.
\end{proof}

For a tree-decomposition $(T,\chi)$ of $G$ and a tree node $t$, the \emph{marginal graph} $G^\star_t$ at $t$ is the graph 
obtained from $G[\margin(t)]$ by adding an edge $uv$ whenever there is a child $b$ of $t$ such that $u,v\in \adh(b)$ (and delete edges if parallel edges are created).

\begin{lemma}\label{lem:torsofiedelim}
Let $(T,\chi)$ be a tree-decomposition of a connected graph $G$
and $\Lambda$ be a $k$-colorable guidance system which captures all bags of $(T,\chi)$.
Then for any vertex $r$ of $G$, there is an elimination tree $F$ of $G$ rooted at $r$ 
whose width is the same as the width of $(T,\chi)$ 
and an oriented path system of congestion $k$ witnessing $F.$
\end{lemma}
\begin{proof}
Let $T$ be rooted at a node whose bag contains $r$. 
Notice that one can modify a tree-decomposition into a sane tree-decomposition without increasing the width 
so that each bag of the sane tree-decomposition belongs to the bag of some node of the input tree-decomposition by~\cref{lem:optimalsane}.
Therefore, we may assume that $(T,\chi)$ is a sane tree-decomposition while maintaining the assumption 
that $\Lambda$ captures all bags of $(T,\chi).$ In particular $\Lambda$ captures all adhesions of $(T,\chi).$

Let $G^\star$ be the graph obtained by 
adding an edge between any pair $u,v$ of vertices whenever there is a tree node $t$ with $u,v\in \adh(t).$
We claim that for each tree node $t$, $\margin(t)$ is connected in $G^\star$, or equivalently the marginal graph $G^\star_t$ is connected. Indeed, suppose that $G^\star_t$ has at least two connected components, and let $U$ and $W$ 
 be a pair of components of $G^\star_t$, as vertex subsets of $\margin(t)$, whose distance in $G[\comp(t)]$ is minimum. As $\comp(t)$ is connected by saneness of $(T,\chi)$, there is an $(U,W)$-path $P$ in $G[\comp(t)]$ and we choose a shortest such path $P$.
  By the choice of $P$ and since $(U,W)$ are components of $G^\star_t$ minimizing their distance in $G[\comp(t)]$, no internal vertex of $P$ is in $\margin(t)$.
  Since the internal vertices of $P$ are contained in $\comp(t) \setminus \margin(t)$ each internal vertex of $P$ is contained in $\comp(b)$ for some child $b$ of $t$. 
  By definition, there is no edge between $U$ and $W$ so $P$ must contain internal vertices.
  Hence, there is some child $b$ of $t$ such that $P$ internally traverses $\comp(b)$.
  By definition of tree decomposition, $\adh(b)$ separates $\margin(t)$ from $\comp(b)$.
  Hence, $P$ must contain some vertex $a \in \adh(b)$.
  By the choice of $U$ and $W$, the only possible vertices in $V(P) \cap \adh(b)$ are the endpoints of $P$.
 Moreover, since the ends of $P$ are not adjacent it follows that $a$ is the unique vertex of $\adh(b)$ in $V(P)$.
So, $V(P) \setminus \{a\}$ is contained in $\comp(b)$, a contradiction since the two endpoints are both in $\margin(t)$. 

\medskip

Now, we construct a collection of spanning trees of $G^\star_t$ over all tree nodes $t$ as follows. 
When $t$ is the root node $T$, it contains the vertex $r$ (the arbitrarily chosen vertex in the main statement) is in the margin. We construct a DFS-tree $F_t$ of $G^\star_t$ rooted at $r.$
Next, choose a tree node $t$ with a parent $t'$ such that a DFS-tree $F_{t'}$ is already constructed, 
and let $w_t$ be the vertex in $\adh(t)\setminus \adh(t')$ which comes the last in the discovery order of $F_{t'}.$ 
Since all vertices of $\adh(t)\setminus \adh(t')$ belong to the margin of $t'$, the discovery order of $F_{t'}$ 
defines a total order on $\adh(t)\setminus \adh(t')$. Moreover, $\adh(t)\setminus \adh(t')$ is non-empty 
by the first property of Lemma~\ref{lem:saneness}. Therefore, such $w_t$ is well-defined.
Now, select an arbitrary neighbor $r_t$ of $w_t$ in $G^\star_t$ contained in the margin of $t.$ 
Such $r_t$ exists dues to the second property in Lemma~\ref{lem:saneness}. Starting with $r_t$ as the root, we construct a DFS-tree $F_t$ of $G^\star_t$. 

Observe that the DFS-trees form a spanning forest of $G^\star$ as the margins of the tree nodes partition $V(G^\star)$. 
We combine the DFS-trees over all tree nodes of $T$ by adding the set of edges $\{(w_t,r_t):\text{$t$ is a non-root node of $T$}\}$ and let $F$ be the resulting rooted tree. From the construction, it is clear that $F$ is a spanning tree of $G^\star$ rooted at 
$r\in \margin(\mathsf{root})=\chi(\mathsf{root})$ and especially $F$ is a subgraph of $G^\star.$
The next claim is clear from the construction of $F$ by the induction on the depth of node $t.$
\begin{claim}\label{claim:ancestor}
Let $t$ be a tree node of $F.$ Any vertex in $\adh(t)$ is an ancestor of $\comp(t)$ in $F.$
\end{claim}

Now we want to prove that $F$ is an elimination tree of $G^\star$ of width at most $\width(T,\chi),$ 
which implies that $F$ is an elimination tree of $G.$
We first argue that  $F$ is an elimination tree of $G^\star$. Consider an arbitrary edge $uv$ of $G^\star.$
Let $t_u$ and $t_v$ be the tree nodes such that $u\in \margin(t_u)$ and $v\in \margin(t_v)$. 
If $t_u=t_v$, then $F$ induced by the vertex set $\margin(t_u)$ is a DFS-tree of $G^\star_{t_u}$ ensures that $u$ and $v$ are comparable in $F$, that is, one is an ancestor of the other. 

Suppose that $t_u$ is a strict ancestor of $t_v$ in $T$ and let $t$ be the child of $t_u$ 
on the unique path from $t_u$ to $t_v$ (possibly with $t=t_v$). Because $u$ and $v$ are adjacent, 
it holds that  $u\in \adh(t)\cap \adh(t_v)$. As $v$ is in the component of $t$, Claim~\ref{claim:ancestor} 
implies that $u$ is an ancestor of $v$ in $F$. 
It remains to observe that $t_u$ and $t_v$ must be comparable in $T$ since otherwise 
the tree nodes of $T$ containing $u$ and those containing $v$ are disjoint, violating the condition 
of a tree-decomposition.

To see that the width of $F$ is at most $\width(T,\chi)$, consider an arbitrary vertex $v$ of $F$ 
and let $t$ be the tree node of $T$ such that $v\in \margin(t).$ By the construction of $F,$  
all vertices of $F_v$, the subtree of $F$ rooted at $v$, is fully contained in $\comp(t).$ 
Moreover, any ancestor of $v$ in $F$ 
belongs to the margin of a tree node $t'$ which is an ancestor of $t$ in $T.$ Therefore, if $w$ 
is an ancestor of $v$ in $F$ and has a neighbor in $F_v,$ $w$ must be contained in $\chi(t).$ 
It follows that $\str(v)\subseteq \chi(t)$, and the claim follows.

Lastly, we demonstrate how to construct a desired path system $\P$ that witnesses the elimination tree $F$ 
of congestion $k$ from the $k$-colorable guidance system $\Lambda.$ Consider a vertex $v$ which is not a root of $F.$ 
If $v$ and $\parent(v)$ are adjacent in $G$, there is nothing to do. 
If this is not the case, then $v$ and $\parent(v)$ are simultaneously in an adhesion of a tree node, say $t$. 
This is because $F$ is a subgraph of $G^*$ and two vertices are adjacent in $G^*$ but not in $G$ 
only if they belong to some adhesion. 
As $\Lambda$ captures the collection of all adhesions of $(T,\chi)$, there is a vertex $x$ of $G$ 
such that $\adh(t)\subseteq \roots_\Lambda(x).$ 
We choose two in-trees of $\Lambda$ rooted at $v$ and $\parent(v)$, respectively.
We add to $\P$ an oriented path contained in the oriented walk from $v$ to $\parent(v)$, where the oriented walk 
is obtained by reversing the orientation of edges of an oriented walk from $x$ to $v$  in one in-tree  and then by concatenating it with the oriented walk from $x$ to $\parent(v)$ in the other in-tree.
The path system $\P$ is constructed by repeating this procedure for every non-root vertex of $F.$

We argue that $\P$ has congestion at most $k.$ 
For any vertex $y$ of $G$, there is at most one path in $\P$ 
which has $y$ as the starting vertex, namely the path connecting $y$ to its parent $\parent(y)$ in $F.$ 
Let $\P_y\subseteq \P$ be the collection of paths in $\P$ traversing $y$ as an internal vertex. 
Clearly, any path $P$ of $\P_y$ is contained in 
the union of all in-trees of $\Lambda$ traversing $y$ and not taking $y$ as any of the roots. 
Therefore, such in-trees have pairwise distinct colors in a $k$-coloring of $\Lambda$ 
and thus there are at most $k$ such in-trees.
Let $\lambda_1,\ldots , \lambda_\ell\in \Lambda$ be a minimal set of in-trees 
whose union contains all paths of $\P_y$, and
let $s_i$ be the root of in-tree $\lambda_i$ for all $i\leq \ell\leq k.$
By construction of $\P$, we add a new oriented path to $\P$ 
connecting  one vertex in $\{s_1,\ldots, s_\ell\}$ to another vertex in $\{s_1,\ldots, s_\ell\}$ 
only when the latter vertex is the parent of the former vertex in $F.$ Therefore, at most $\ell-1\leq k-1$ 
paths are created from the in-trees $\lambda_1,\ldots , \lambda_\ell$ and added to $\P.$ 
This completes the proof that the congestion of $\P$ is at most $k.$
\end{proof}

There is a limitation in applying Lemma~\ref{lem:torsofiedelim} to an arbitrary tree-decomposition 
as such a guidance system might not exist for every tree-decomposition. 
For example, consider a cycle $C_n$ of length $n$ and a path-decomposition of $C_n$ in which each bag is a `triangle' so that the resulting triangulation makes a triangulated ladder. One can quickly  convince oneself that any guidance system capturing all adhesions of this tree-decomposition requires $\Omega(n)$ colors. In contrast, a path-decomposition of $C_n$ in which every bag is in the form $\{v, x_i, x_{i+1}\}$ for some vertex $v$ and an edge $x_ix_{i+1}$ of  the path $C_n-v$, on the vertices $x_1,x_2,\cdots ,x_{n-1}$ in order, admits a guidance system of low chromatic number as observed also in~\cite{BojanczykP2016}. Indeed, let $\vec{e_i}$ be the trivial in-tree obtained by orienting the edge $x_ix_{i+1}$ from $x_i$ and $x_{i+1}$ and let $\vec{P}$ be the in-tree (directed path) obtained from the path $C_n-vx_1$ by taking $v$ as the root. Then the guidance system $\Gamma=\{\vec{P}\}\cup \{\vec{e_i}\mid i\in [n-2]\}$ can be colored using three colors and it certainly captures all adhesions. 
This indicates that one may need to locally change the given tree-decomposition in order to construct a guidance system with a low chromatic number.  

The next lemma presents how to combine elimination trees when such a `local surgery' is required.
A stronger form of guidance system capturing the adhesions of a tree-decomposition is useful in this context. 
Given a (rooted) tree-decomposition $(T,\chi)$ and a guidance system $\Lambda$ 
which captures the adhesions of $(T,\chi)$, we say that $\Lambda$ \emph{strongly captures} the adhesions of 
$(T,\chi)$ if for every non-root tree node $t$ of $T$, there is a vertex $s_t$ in $\margin(t)$ 
which 
\begin{enumerate}
\item[(i)] captures $\adh(t)$, i.e., $\adh(t)\subseteq \roots_\Lambda(s_t)$, and 
\item[(ii)] for every vertex $v\in \adh(t)$, there is an $(s_t,v)$-path in an in-tree of $\Lambda$ rooted at $v$ 
whose vertices are fully contained in $\comp(t)$ except for the destination vertex $v.$
\end{enumerate}

 \begin{lemma}\label{lem:combineelim}
 Let $(T,\chi)$ be a sane tree-decomposition of a connected graph $G$ such that the following holds for some $k_1, k_2$ and $\width$:
 \begin{itemize}
 \item for every node $t$ and an arbitrary vertex $r_t\in \margin(t)$, the marginal graph $G^\star_t$ at $t$ 
 admits an elimination tree $F_t$ of width $\width$ rooted $r_t$
 and an oriented path system $\P_t$ over $G^\star_t$ of congestion $k_1$ witnessing $F_t$, and
 \item the family of adhesions of $(T,\chi)$ can be strongly captured by $k_2$-colorable guidance system $\Lambda$ 
 over $G.$
 \end{itemize}
 Then $G$ admits an elimination tree $F$ of width at most $\width+\max_{t\in V(T)} \abs{\adh(t)}$ 
 and a path system $\P$ over $G$ witnessing $F$ with congestion $3k_1\cdot k_2^2.$
 \end{lemma}
 \begin{proof}
Let $G^\star$ be the graph obtained from $G$ by 
adding an edge between any pair $u,v$ of vertices whenever there is a tree node $t$ with $u,v\in \adh(t).$
We note that each $G_t^\star$ is an induced subgraph of $G^\star.$

First, we construct an elimination tree for each marginal graph and a path system witnessing the elimination tree 
in a top-to-down manner. 
Due to the first precondition in the main statement, it suffices to choose the root of each elimination tree.
When $t$ is the root node of $T,$ we choose an arbitrary vertex of $\margin(t)=\chi(t)$ as the root of $F_t^\star.$ 

Consider a node $t$ such that for its parent $\parent(t)$ a pair $(F_{\parent(t)},\P_{\parent(t)})$ of 
an elimination tree $F_{\parent(t)}$ of $G_t^\star$ and a path system over $G_t^\star$ has been constructed. 
Recall that $\adh(t)\setminus \adh(\parent(t))\neq \emptyset$ by Lemma~\ref{lem:saneness}.
Moreover, it is easy to see that $\adh(t)\setminus \adh(\parent(t))\subseteq \margin(\parent(t))$ holds, 
which implies that $\adh(t)\setminus \adh(\parent(t))$ forms a clique in the marginal graph $G_{\parent(t)}^\star$. 
Therefore, the vertices of $\adh(t)\setminus \adh(\parent(t))$ are totally ordered by the elimination tree $F_{\parent(t)}$ 
at the parent $\parent(t)$ of $t.$ Let $w_t$ be the vertex which comes the last in the total order, or equivalently, 
the lowermost one in $F_{\parent(t)}$ amongst the vertices of $\adh(t)\setminus \adh(\parent(t))$. 
Let $r_t$ be a neighbor of $w_t$ in the graph $G^\star_t$ that belongs to $\margin(t)$. By Lemma~\ref{lem:saneness},
such $r_t$ exists and this vertex will be the root of an elimination tree $F_t.$

Now the elimination tree $F$ of $G^\star$ is obtained from the (disjoint) union of $F_t$ over all tree nodes $t$ of $T$ 
by adding an edge between $w_t$ and $r_t$ for every non-root node $t.$ It is routine to verify that $F$ is indeed an elimination tree of $G^\star$, and thus an elimination tree of $G.$ To verify that the width of $F$ is within the claimed bound, consider an arbitrary vertex $v$ of $G$ and let $t$ be the tree node whose margin contains $v.$ Any ancestor $v'$ 
of $v$ in $F$ 
which is not taken into account in the width of $F_t$ and has a neighbor in the subtree $F_v$ of $F$ rooted at $v$,
is contained in the margin of a tree node $t'$ which is a strict ancestor of $t$ in $T.$ As all vertices of $F_v$ are 
entirely contained in $\comp(t)$, $v'$ must be in $\adh(t)$ were it to be adjacent with any vertex of $\comp(t)$ and 
thus $\str_F(v)\setminus \str_{F_t}(v)\subseteq \adh(t)$. By the precondition  
$\abs{\str_{F_t}(v)}\leq \width$, it follows that the width of $F$ is at most $\width+\max_{t\in V(T)} \abs{\adh(t)}$.

Second, let us construct a path system  over $G$ which witnesses $F$ using $\P_t$ and $\Lambda.$
For each node $t$ of $T$, we denote by $A_t$ the set of edges which are present in the marginal graph 
$G^\star_t$ and not in $G[\margin(t)].$ We distinguish the tree edges of $F$ depending on whether it is an edge of $F_t$ 
for some tree node $t$, or it is of the type $(w_t,r_t)$ for some tree node $t$ which connects the root of $F_t$ 
with some vertex of the elimination tree $F_{\parent(t)}$ of the marginal graph $G_{\parent(t)}^\star$, where $\parent(t)$ is the parent of $t$ in $T.$ We call the former an in-type and the latter ex-type. 

\medskip

\noindent {\bf ex-type edge of $F$.} For an edge $(w_t,r_t)$ of $F$ which is not present in $G$, recall that $(w_t,r_t)$ is present as an edge of $G^\star$ as $F$ is a spanning tree of $G^\star$. In particular, there is a child $t'$ of $t$ such that $w_t,r_t\in \adh(t')$ by construction of $G^\star.$ As $\Lambda$ captures all adhesions of $(T,\chi)$, one can find an oriented path from $r_t$ to $w_t$ in the union of two in-trees of $\Lambda$, rooted at $r_t$ and $w_t$ respectively, having a common vertex by reorienting some edges if necessary. The existence of such in-trees is ensured by the fact that $\Lambda$ captures all adhesions of $(T,\chi).$

\medskip

\noindent {\bf in-type edge of $F$.} Consider an edge $(x,y)$ of $F$, with $y=\parent_F(x)$, which is not of the type $(w_t,r_t)$, and let $t$ be the tree node of $T$ such that 
$(x,y)$ is an edge of $F_t.$ As $\P_t$ witnesses $F_t$, there is an $(x,y)$-path $P_{xy}$ in $\P_t.$ 
As $\P_t$ is a path system over $G_t^\star$, some edges of $P_{xy}$ may not be present in $G.$ 
Now, whenever an oriented edge $(u,v)$ appears in $P_{xy}$ with $uv\in A_t$, 
we know that there is a child $t'$ of $t$ in $T$ with $u,v\in \adh(t').$ We replace $(u,v)$ by a suitable oriented path 
contained in the union of two in-trees of $\Lambda$ rooted at $u$ and $v$ respectively. Especially, we replace $(u,v)$ 
by the concatenation of two paths, namely $(s_{t'},u)$-path and $(s_{t'},v)$-path which are fully contained in 
$\comp(t')$ except for $u$ and $v$. Such paths exist due to the assumption that $\Lambda$ strongly captures the adhesions of $(T,\chi)$. 
After replacing all occurrences 
of edges of $G_t^\star$ not present in $G$ with a suitable oriented path obtained from $\Lambda$, we end up 
in an oriented walk from $x$ to $y$ in $G.$ Within this walk, one can find an oriented $(x,y)$-path. 

\bigskip

Let $\P_{ex}$ be the  path system obtained for the ex-type edges $(r_t,w_t)$ over all non-root tree nodes $t$ of $T$, 
and let  $\P_{in}$ be the path system obtained for the in-type edges. 
The output path system $\P$ is the union of $\P_{ex}$ and $\P_{in}.$ It is clear from the construction that 
$\P$ is a path system over $G$ which witnesses $F.$ It can be easily verified using the same argument as in Lemma~\ref{lem:torsofiedelim} that the congestion of $\P_{ex}$ is 
at most $k_2$. 
 
\bigskip

Therefore, it remains to bound the congestion of $\P_{in}.$ Consider an arbitrary vertex $v$ of $G$ and let 
$t$ be the tree node such that $v\in \margin(t)$. We shall examine for which in-type edges $(x,y)$ of $F$ 
the associated replacement $(x,y)$-path added to $\P_{in}$ traverses $v.$ 
Consider an $(x,y)$-path $P_x\in \P_{in}$, where $y$ is the parent of $x$ in $F$, which traverses $v$ as an internal or start vertex and let $t'$ be the tree node such that $(x,y)$ is an edge of the elimination tree $F_{t'}.$ 

It is not difficult to see that $t'$ is an ancestor of $t$, possibly with $t'=t.$ Indeed, we construct $P_x$ 
from an $(x,y)$-path $P^\star_x$ of $G_{t'}^\star$, then replace some edge of $P^\star_x$ from $A_{t'}$ 
by substituting it with a path from $\Lambda.$ As $\Lambda$ strongly captures $\adh(t')$, the path used from 
substituting an edge of $A_{t'}$ is fully contained in the component of a child $t''$ of $t'$, save the two endpoints 
(which are in $\adh(t'')$). 
Therefore, every vertex of $P_x$ is in $\comp(t').$ In particular, if $P_x$ traverses $v$ as a start or internal vertex,
$v$ must belong to $\comp(t')$ and $t$ is a descendant of $t'.$

There are two possibilities. If $t'=t$, then $P^\star_x$ (that is, before the replacement of edges of $A_t$) traverses $v$ 
as an internal or start vertex already. Therefore, there are at most $k_1$ such paths in $\P_{in}$ 
as the congestion of $\P_t$ is upper bounded by $k_1$. 

Now, consider the case when $t'$ is a strict ancestor of $t.$ Let $t_1,\ldots , t_\ell$ be the set of strict ancestors of $t$ 
such that for each $i\in[\ell]$, there is an in-type edge of $F_t$ whose replacement path traverses $v.$ We want to argue that  $\ell\leq k_2/2$ holds. For this, each $i\in [\ell],$ 
choose a path $R_i$ traversing $v$ which is obtained from two in-trees of $\Lambda$ 
whose roots form an edge $(a_i,b_i)$
of $A_{t_i}.$ As $a_i,b_i$ is contained in $\margin(t_i)$, all the in-trees involved in $\{R_i\mid i\in[\ell]\}$ have distinct roots. 
This, together with the fact that all $R_i$'s share a common vertex $v$, implies that all these in-trees have distinct colors, and thus $\ell\leq k_2/2.$ 

Now choose a strict ancestor $t'$ of $t$ for which the  largest number of replacement paths $P_x$ traversing $v$ 
has been added to $\P_{in}.$ Let $t''$ be the child of $t'$ on the unique $(t',t)$-path of $T$ and note that 
$\abs{\adh(t'')}\leq \abs{\roots_\Lambda(s_{t''})}\leq k_2$, where $s_{t''}$ is a vertex in $\comp(t'')$ such that 
$\adh(t'')\subseteq \roots_\Lambda(s_{t''})$. Therefore, there are at most $k_1\cdot k_2$ paths in the 
path system $\P_{t'}$ over $G^\star_{t'}$ which traverses some vertex of $\adh(t'').$ Clearly, 
if any replacement path $P_x$ for some in-type edge of $F_{t'}$ traverses $v$, such $P_x$ must traverse 
a vertex of $\adh(t'')$ as an internal vertex or as a starting vertex. Hence, we conclude that $t'$ contains at most $k_1\cdot k_2$ 
replacement path traversing $v.$

To sum up, the congestion of $\P=\P_{ex}\cup \P_{in}$ is bounded by $k_2+k_1 + \frac{k_1\cdot k_2^2}{2}\leq 3k_1\cdot k_2^2.$
\end{proof}

\bigskip

Finally, a tree-decomposition as in~\cref{lem:combineelim} is presented in the work of Bojańczyk and 
Pilipczuk~\cite{BojanczykP2016}, summarized in the next two statements\footnote{~\cref{lem:globalcapture} 
was originally in Lemma 5.1 of \cite{BojanczykP2016} without the phrase `strong'. The presented version
can be easily obtained using Lemma 5.9 of \cite{BojanczykP2016} though.}.

\begin{lemma}\cite[Lemma 5.1+Lemma 5.9]{BojanczykP2016}\label{lem:globalcapture}
Let $G$ be a graph of treewidth at most $\width.$ 
Then $G$ admits a sane tree-decomposition $(T,\chi)$ such that the following holds.
\begin{enumerate}
  \item every marginal graph $G^\star_t$ has path-width at most $2\width+1$, and
\item the family of adhesions of $(T,\chi)$ can be strongly captured by a guidance system colorable with $4\width^3+2\width.$
\end{enumerate}
\end{lemma}

\begin{lemma}\cite[Lemma 4.4]{BojanczykP2016}\label{lem:localcapture}
Let $G$ be a graph of path-width at most $\width.$ 
Then there is a tree-decomposition $(T,\chi)$ of $G$ and a guidance system $\Lambda$ of $G$ 
such that
\begin{itemize}
\item $\Lambda$ captures $\{\chi(t):\text{$t$ is a node of $T$}\}$, and 
\item $\Lambda$ is $f(\width)$-colorable for some function $f\in 2^{2^{O(\omega^2)}}$.
\end{itemize}
\end{lemma}

Observe that $f(\width)$-colorability of the guidence system $\Lambda$ given by \cref{lem:localcapture} is an upper-bound on the width of $(T,\chi)$.
Indeed, each bag of $(T,\chi)$ is captured by some vertex $u$, and therefore, the vertices in a bag are the roots of  in-trees in $\Lambda$ all of which have pairwise distinct colors. 
Now we are ready to complete the proof of~\cref{thm:tw2lvtw}.

\bigskip

\begin{proofof}{~\cref{thm:tw2lvtw}}
We may assume that $G$ is connected. For a graph $G$ of treewidth at most $\width,$ 
let $(T,\chi)$ be a tree-decomposition of width satisfying the properties described in Lemma~\ref{lem:globalcapture}, 
and let $\Lambda$ be a $(4\width^3+2\width)$-colorable guidance system over $G$ which strongly captures all adhesions of $(T,\chi)$. 
By~\cref{lem:globalcapture}, each marginal graph $G^\star_t$ has pathwidth is at most $2\width+1$.
Hence, one can apply Lemma~\ref{lem:localcapture} to $G^\star_t$ to each tree node $t$, which yields a tree-decomposition of $G^\star_t$ whose bags can be captured  by $f(\width)$-colorable guidance system over $G^\star_t.$
Notice that this decomposition has width at most $f(\width)$. 
By Lemma~\ref{lem:torsofiedelim}, for any vertex $r_t$ of $\margin(t)$, there is an elimination tree $F_t$ of $G^\star_t$ rooted at $r_t$ 
of width $f(\width)$ as well as a path system $\P_t$ over $G^\star_t$ which witnesses $F$ with  congestion $f(\width)$.
Now applying Lemma~\ref{lem:combineelim} leads to an elimination tree $F$ of $G$ of width at most $f(\width)+\max_{t\in V(T)} \abs{\adh(t)}$ 
and a path system $\P$ over $G$ which represents $F$ with congestion $3f(\width)\cdot (4\width^3+2\width)^2.$

To bound $\max_{t\in V(T)} \abs{\adh(t)}$, we claim that $\abs{\adh(t)}$ cannot exceed the chromatic number of $\Lambda$.
Indeed, for each bag $\chi(t)$ of a tree node $t$, there is a vertex $u$ such that
$\adh(t)\subseteq \roots_\Lambda(u)$. Notice that we may assume there is no pair of in-tree in $\Lambda$ 
which share a root and are vertex-disjoint save the root because 
if such a pair, they can be replaced by a single in-tree without increasing the chromatic number of $\Lambda.$ 
Therefore, all the in-trees containing $u$ pairwise intersect and thus have different colors in 
a coloring of $\Lambda$. As $\Lambda$ is $(4\width^3+2\width)$-colorable, we conclude that 
the width of $F$ is at most $f(\width)+\max_{t\in V(T)} \abs{\adh(t)}\leq f(\width)+ (4\width^3+2\width)$.
This completes the proof.
\end{proofof}

\section{\mso \ model-checking}\label{sec:mso}
\newcommand{\Pp}{\mathcal P}

To construct the model-checking certification in \cref{thm:msoPLS}, we build on the certification presented in \cref{prop:locallyPLS}.

We show how to reduce \cref{thm:msoPLS} into an analogous statement for \msoo\ model-checking. %
This follows a standard approach of converting \mso\ to \msoo\ by extending the universe to contain vertices as well as edges, subject to the appropriate labeling that matches the origin of elements in the universe.
Then it is relatively straightforward to convert a given \mso-formula to an equivalent  \msoo-formula.
The typical final nail to this argument is that the subdivision of a graph does not change its treewidth.
In general, such a conversion works for sparse graph classes; see \cite{CourcelleE2012} for more information on the topic.
However, we cannot use this as a final argument right away in the case of local certification.
Still, we will mimic the suggested scheme as we distribute supposed certificates for edges to incident vertices based on the bounded degeneracy of the graph.
We formalize this argument as follows.

\begin{observation}\label{obs:MSO2toMSO1}
  Let $d$ be a fixed integer and $G$ be a $d$-degenerate graph.
  Then we can translate an instance of local certification on $G$, where both the certification and the verification take place also on edges of $G$ into an instance of local certification where certificates of the edges of $G$ are distributed to vertices of $G$ in such a way that each vertex receives additionally at most $d$ certificate while the verification occurs only at vertices of $G$. 
\end{observation}

\begin{proof}
  Distribute edge certificates based on a degeneracy order of $G$.
  Then each vertex receives at most $d$ edge certificates and can verify them as each vertex receives or has the certificates the respective edge would have received and verified on its own previously.
\end{proof}

Hence, we will derive the proof of \cref{thm:msoPLS} as a direct combination of \Cref{obs:MSO2toMSO1} and the following theorem for \msoo (\Cref{thm:MSO1mc}) via the discussion above. 
We provide a formal proof of the discussion above after the statement of the main theorem.

\begin{theorem}[Local certification of \msoo\ model-checking]\label{thm:MSO1mc}
For every integer $k$ and for every \msoo-sentence $\psi,$ there is a proof labeling scheme $(\prover, \verifier)$ of size $O(\log n)$ 
for the class of connected graphs that has locally verifiable treewidth at most $\width$ 
and satisfies the sentence $\psi.$
\end{theorem}

\begin{proof}[Proof of \Cref{thm:msoPLS}]
We convert an \mso-formula into \msoo-formula using two additional predicates $L_V$ and $L_E$ representing whether those elements were originally vertices or edges, respectively.
Then, we take one relation representing the incidence.
Then, whenever there was an edge in the original formula, we now require an element of the universe with the label $L_E$ and analogically for the case of former vertices.
It is clear that the two formulas are equivalent, but this conversion extended the universe to contain edges as well as vertices.
However, we now use \cref{obs:MSO2toMSO1} to simulate certification on edges in the respective vertices as a graph of bounded treewidth has bounded degeneracy.
\end{proof}

Before we dive into the proof of the core argument of this section we need to introduce some model-checking notation and explain the main tool we will be using for the local certification.
Our approach is based on a proof published in the lecture notes\footnote{This approach originates in some unpublished work of Hliněný, Král’, and Obdržálek.}~\cite[Section 5.4]{Fiala} describing how to obtain classical \msoo\ model-checking for graphs of bounded treewidth reproving the seminal Courcelle's theorem~\cite{Courcelle90twDP}.
The approach uses objects called \say{evaluation trees} which play the role of remembering a structure of all possible assignments of vertices of a graph to variables of a formula.
The advantage of the evaluation trees is that in each node of the elimination tree, we can deal with a single object, the evaluation tree, that we modify bottom up on the elimination tree.
We only evaluate it in the root {\sl after} we have seen the whole graph to obtain the result of the model-checking.
The evaluation tree technique was recently formalized for the cliquewidth~\cite{DM24} and we will follow the notation established in this paper.
Due to the robustness of this technique, we are able to separate the model-checking part into a self-contained statement and use it as a black box for the local certification procedure.
Therefore, we present at this point only the necessary terminology before showing the proof of the local certification part.
For ease of discussion and without loss of generality, we assume for the whole paper that the formula $\psi$ in question is the prenex normal form. 
\begin{definition}[Evaluation tree]\label{def:evalTree}
  The  \emph{(full) evaluation tree} for a sentence $\Phi$ with $q$ quantifiers and a graph $G$ is a rooted tree $Q$ with $q + 1$ levels constructed as follows.
  We define the first level to be a single node.
	For each $1 \geq \ell \geq q$ we define the $(\ell +1)$-th level of the evaluation tree recursively as follows:
  Let $z_{\ell +1}$ denote the $\ell$-th variable in the formula $\Phi$. 
  If $z_{\ell + 1}$ is an individual variable, then $w$ has $n$ children, representing all possible choices of vertices for $z_{\ell + 1}$. 
  If $z_{\ell + 1}$ is a set variable, then $w$ has $2^{\abs{V(G)}}$ children, representing all possible choices of subsets of vertices for $z_{\ell + 1}$. 
\end{definition}
The answer whether $\Phi$ is satisfied on $G$ can be obtained by traversing $Q$ bottom-up. 
  Indeed, each leaf of $Q$ is associated with an interpretation $s$ as an assignment of the variables to vertices and vertex sets. 
  Hence, one can evaluate the quantifier-free formula in the suffix of $\Phi$ and determine the \true\ or \false\ value.
The value of each internal node depends on the values of its children only.
In case a node corresponds to a universally quantified variable, it must have all its children satisfied.
At least one child satisfied is enough for an existentially quantified variable.
We refer to the described procedure as \emph{evaluation} of the evaluation tree.

We can observe that the evaluation tree contains all the information needed to successfully evaluate a given formula on a given graph (as a relational structure), but a big drawback is that it is astronomically large.
However, the plan is to reduce the size of the tree so that will contain information about only $\width+q$ vertices. 
The main idea for such a reduction comes from the easy observation that in order to evaluate a leaf of $Q$ we only need to know the adjacency of vertices assigned by at least one vertex (individual) variable as well as the {\sl labels} of those vertices.
Here, a set variable whose image under the interpretation $s$ contains a vertex $v$ is treated as a label of $v$ and a vertex may have multiple labels. 
Obviously, this is great as the number of such vertices is at most $q$.
However, we need to store a bit more information to be able to maintain such information throughout the induction in the elimination tree. 
For that, we store additional information about vertices in $\str(v)$ where $v$ is the vertex representing a node in the elimination tree.
Moreover, we need to remember which vertex variables will be assigned in a subgraph corresponding to a different branch of the elimination tree.
Altogether, this results in the promised bound on the size of the evaluation tree as a function of depending solely on $\tw(G)+q$.
We call such a reduced object a \emph{compact evaluation tree} and postpone the formal definition (\cref{def:compactTree}) to \cref{sec:mcProof}.

We now state the main model-checking theorem with specific properties that will be useful for the local certification proof.
We dedicate \cref{sec:mcProof} to its proof.

\begin{theorem}\label{thm:mcMSO1onEF}
  Let $\phi$ be a fixed \msoo-sentence using $q$ quantifiers and let $\width$ be an integer.

  Then there exist a computable function $f$ such that for any graph $G$ together with its elimination tree $F$ of width $\width$, and for each node $v$ of $F$, there is a compact evaluation tree $Q_v$ that satisfies all the following properties:
\begin{itemize}
  \item[0)] for the root $r$ of $F$, the compact evaluation tree $Q_r$ evaluates to true if and only if $G\models \phi$, \phantomsection\label{def:0}
\item[1)]\phantomsection\label{def:1} for each $v\in V(F)$, the size of $Q_v$ is bounded by $f(q+\width)\cdot \log(n)$.
 More precisely,
 $Q_v$ has at most $f(q+\width)$ many nodes while in the leaves of $Q_v$ we used at most $\log n$ long labels to encode the name of vertices in $\str(v)$\footnote{The bound on the bit size of $Q_v$ could be improved to just $f(q+\width)$. However, it would require more complicated proof and such an improvement does not improve the overall complexity of the main statement of this paper, so we decided to present the simpler weaker version of this bound.} whose number is bounded by $|\str(v)|\le\width$,
  \item[2)] for each $v\in V(F)$, $Q_v$ can be computed uniquely from: 
    \phantomsection\label{def:2}
    \begin{itemize}
      \item $Q_{u_1}, Q_{u_2},\ldots$ where $u_1,u_2\ldots$ are all children of $v$ in $F$ in some fixed order, and
      \item the full adjacency information between $v$ and $\str(v)$. %
    \end{itemize}
\end{itemize}
\end{theorem}

Note that in the statement above Property \hyperref[def:0]{0)} alone is enough to prove the above-mentioned Courcelle's theorem~\cite{Courcelle90twDP}.

\subsection{Proof of \Cref{thm:MSO1mc}}

Here we prove the core of local certification model-checking result using \cref{thm:mcMSO1onEF}.

\begin{proof}[Proof of \Cref{thm:MSO1mc}]
First, we invoke \cref{prop:locallyPLS} to obtain the certified structure of elimination tree $F$ together with the respective path system $\Pp$ that we will use to certify the formula locally.

  We describe certificates that prover $\prover$ adds to certificates described by \cref{prop:locallyPLS}.
  For each vertex $v$, we store an $O(\log n)$-sized certificate $C_v$ that is supposed to play a role of the evaluation tree at node $v$ of the elimination tree.
  Additionally, we add more certificates to the cargo.
  For each $(u,v)$-path $P$ in $\Pp$, we add $C_u$ as cargo to each vertex of $P$ except for the last one.
  Due to properties of $\Pp$, in particular the congestion bound on $\Pp$,  we store $\width$ times $O(\log n)$-sized certificates. 
Therefore, we can ensure that $N_G(v)$ contains certificates with the evaluation trees of all children of $v$ in $F$ and due to previous certification, they are all marked to be used at vertex $v$. 
  We describe the algorithm verifier  $\verifier$ performs at each vertex $v$ of $G$: 
  \begin{itemize}
    \item \textbf{Main verification.}\phantomsection\label{check:main}
 As discussed, we know that $N_G(v)$ contains correct information about who are children of $v$ in the elimination tree together with certificates of their respective evaluation trees all marked to be used at vertex $v$ because that was the information $\prover$ stored in the cargo.
 Hence, $v$ received the communication $\id{}$s of all its children and their proposed evaluation trees in one round of communication.
      By Property \hyperref[def:2]{2)} of \cref{thm:mcMSO1onEF}, each $v \in V(G)$ computes $Q_v$ an evaluation tree for $v$ from the evaluation trees of all the children ordered by the $\id$s of the respective vertices.
      In order to do that, we also need to know full adjacency information between $v$ and $\str(v)$, but as the ancestor information was certified correctly (due to \cref{prop:locallyPLS}); exactly vertices in $\str(v)\cap N_G(v)$ are adjacent to $v$ and others are not.
      At this point, we compare computed $Q_v$ with certificate $C_v$ and reject if they differ.
   
 \item \textbf{Auxiliary verifications.}\phantomsection\label{check:aux}
   Those verifications will be the same as in \cref{prop:locallyPLS}.
   The only goal is to ensure that the whole path $P\in\Pp$ contains identical information.
   Consider an $(u,v)$-path $P\in\Pp$.
   Vertex $u$ verifies that its cargo is identical to $Q_v$ and whether it is marked to be used at $v$.
Any other vertex on $P$ except $v$ verifies whether its cargo is exactly the same as the cargo of its predecessor and successor.
If any such check fails, the corresponding vertex rejects the instance.

 \item \textbf{Model-checking verification.}\phantomsection\label{check:mc}
   In case $v$ is a root of $F$, the verifier $\verifier$ evaluates the evaluation tree $Q_v$ and rejects if $Q_v$ is evaluated to false.
  \end{itemize}

  \paragraph{Completeness.}
  Assume $G\models \phi$ and has locally bounded treewidth.
  We combine certification of locally bounded treewidth done in \cref{prop:locallyPLS} with the following.
  We assess that we can use evaluation trees given by \cref{thm:mcMSO1onEF} and use them as certificates $C_v$ for $v\in V(G)$.
  We distribute a copy of $C_v$ along paths in $\mathcal P$ as a cargo.
  By Property \hyperref[def:1]{1)}, the size of evaluation trees is $O(\log n)$.
By Property \hyperref[def:2]{2)}, we know that neither the \hyperref[check:main]{Main verification} nor the \hyperref[check:aux]{Auxiliary verifications} will reject.
  The \hyperref[check:mc]{Model-checking check} will also not reject thanks to Property \hyperref[def:0]{0)} as we assumed graph $G\models \phi$.
 
  \paragraph{Soundness.}
  Assume $G\not \models \phi$, but still no vertex rejected the instance.
  By the combination of the \hyperref[check:aux]{Auxiliary verifications} and the \hyperref[check:main]{Main verification}, we know that $C_v$ is a correct evaluation tree as given by \cref{thm:mcMSO1onEF} for each $v\in V(G)$.
  However, it means that the evaluation tree in the root should evaluate to true in the \hyperref[check:mc]{Model-checking verification} which is a contradiction with Property~\hyperref[def:0]{0)} of \cref{thm:mcMSO1onEF}.
\end{proof}

\subsection{Proof of \msoo\ model Checking on Graphs with an Elimination Forest of Bounded Width}\label{sec:mcProof}

This subsection is devoted to the proof of \cref{thm:mcMSO1onEF}.
Before we dive into the proof we need to formalize some tools.
We start with definitions of various evaluation trees.
A partial evaluation tree allows us an option of saying that some vertex variable is defined in a different subgraph.
We show how to adapt \cref{def:evalTree} to accommodate that.
In what follows $G_t$ is a graph induced by all vertices in the subtree of the elimination tree rooted at $t$.

\begin{definition}[Partial evaluation tree]
  The \emph{partial evaluation tree} for a formula $\Phi$ and a node $t$ of the elimination tree $F$ of $G$ is a tree obtained from the evaluation tree for $\Phi$ and $G_t$ by adding to each node of level $\ell$ where $z_{\ell + 1}$ is an individual variable an extra child labeled by $ext$ representing the case where $z_{\ell + 1}$ does not belong to $G_t$. 
  We say the value of $z_{\ell+1}$ is \emph{external} in that case.
\end{definition}

The goal will be to identify the isomorphism classes of subtrees of the partial evaluation tree, so we can reduce its size while keeping all the necessary information that we need to evaluate the formula.
In order to do it, we formulate what are the properties we can express about the leaves which allows us to forget some information while preserving the logical value of the formula. %
The aim is to create definitions such that the logical value of isomorphic subtrees is identical.
\begin{definition} [Configuration]\label{def:config}
Let $F$ be an elimination tree of width $\width$, $Q_t$ be a partial evaluation tree for a formula $\Phi$ with $q=q_v+q_S$ variables, and $t \in V(F)$.
A~\emph{configuration} of a leaf is a tuple consisting of: %
\begin{itemize}
  \item a subgraph $G_t'$ of $G[V(G_t)\cup \str(t)]$ induced by at most $q_v$ vertices where some pair of vertices might have their edge relation undefined %
  \footnote{Note that we care only about the shape of such a subgraph and not about any specific vertex names.},
\item an assignment of values from $V(G_t') \cup \{ext\}$ to the individual variables such that each vertex of $V(G_t')$ is assigned to some individual variable at least once,
  \item an assignment of subsets of $V(G_t')$ to set variables, %
\end{itemize}
This tuple satisfies the following properties:
\begin{itemize}
  \item non-external vertices of $G_t'$ could be additionally marked by labels corresponding to vertices in $\str(t)$, but those markings are disjoint (i.e., no two vertices of $V(G_t')$ obtain the same mark and no vertex of $V(G_t')$ obtains two such marks.),
  \item for each pair of vertices in $G_t'$, we know whether there is an edge between them or not unless both are marked by a vertex in $\str(t)\setminus\{t\}$.
\end{itemize}  
\end{definition}
Note that we use labels of size $\mathcal{O}(\log n)$ to describe the markers by vertices in $\str(v)$, but at all times there will be at most $\width$ many such distinct markers.
Observe that for any node $t$ of the elimination tree, configurations of leaves in a partial evaluation tree can be determined.
However, we would like to go further and store in the leafs the configuration and keep on the edges only labels $ext$ and $\str(t)$ in case of vertex variables.
In order to do that, we define a configuration evaluation tree.
 
A \emph{Configuration evaluation tree} of an evaluation tree $F$ rooted at $t$ is a tree where all leaves store a configuration and such that no edges are labeled except for the case of vertex variables and labels $ext$ or $\str(t)$ such that it satisfies the following consistency requirements:
For each subtree $T'$ of depth $\ell<q$ the value of the first $\ell$ variables are consistent in configurations stored among all the leaves of $T'$.
A vertex variable $z$ is \emph{consistent} in the set of configurations with the respective set of assignments $A$ if 
\begin{itemize}
  \item for all $a\in A$, $a(z)=ext$, or
  \item for all $a\in A$, mark of $a(z)$ is the same vertex in $\str(t)$, or
  \item for all $a\in A$, $a(z)$ is not marked and $a(z)\neq ext$.
  \item Additionally, for all variables $z_1,\ldots,z_{\ell'}$ before $z=z_{\ell'+1}$ (in the order given by the prenex normal form), for all $a\in A$, $a(z)$ contains the same subset of labels among $z_1,\ldots z_{\ell'}$.
\end{itemize}
A set variable $Z$ is \emph{consistent} in the set of configurations with the respective set of assignments $A$ if, for all $a\in A$, there is a vertex in $a(Z)$ that is marked by the same vertex in $\str(t)$.
  Additionally, for all variables $z_1,\ldots,z_{\ell'}$ before $Z=z_{\ell'+1}$ (in the order given by the prenex normal form), for all $a\in A$, the set of vertices $a(Z)$ with labels restricted to labels among $z_1,\ldots z_{\ell'}$ is identical.
Note that a configuration evaluation tree can be easily obtained from a (partial) evaluation tree.
That means there is an algorithm converting a partial evaluation tree to a configuration tree.

Now, we can define the isomorphism.
We say that two configurations are \emph{isomorphic} (denoted using $\sim$) if the corresponding graphs are isomorphic including all the labels and markings.
Building on that we define the isomorphism of configuration trees.

\begin{definition}[Isomorphism between configuration evaluation trees]
  Two configuration evaluation (sub)trees $Q$ and $Q'$ of height $h\in[q]\cup \{0\}$ for the same formula $\psi$ and the same graph $G$, %
  are \emph{isomorphic} (also denoted using $\sim$) if 
  \begin{itemize}
\item If $h = 0$, the configuration of $Q$ is isomorphic to the configuration of $Q'$
\item and if $h > 0$, for all children of $Q$ there is an isomorphic child of $Q'$ and vice-versa.
\end{itemize}
\end{definition}

Now, we define our main object which is the reduced evaluation tree because a reduced evaluation tree does not need to be partial.
Hence, we define a compact evaluation tree as a shortcut to capture both properties together.

\begin{definition}[Compact evaluation tree]\label{def:compactTree}
  \emph{Reduced evaluation tree} is a tree where each leaf contains a configuration, and such that no two of its siblings are isomorphic.
For brevity, we say that we \emph{reduced} a configuration evaluation tree.
We denote a configuration evaluation tree that is both partial and reduced as \emph{compact evaluation tree}.
\end{definition}
Note that whenever there is a configuration evaluation tree that is not reduced, we can perform a series of reductions exhaustively, so that the resulting tree is a reduced evaluation tree.

Observe that the number of non-isomorphic compact evaluation trees is bounded solely by a function of $q+\width$.
In particular, the number of non-isomorphic leaves is at most $3^{q_v \choose 2} \cdot  (q_v + 1) ^ {q_v} \cdot (2^{q_v})^{q_s} \cdot (q_v+\width)^{\underline{\width-1}}$, where $q=q_v+q_S$ is the sum of number of vertex quantifiers $q_v$ and set quantifiers $q_S$.
We formalize this observation in the following lemma.

\begin{lemma}\label{lem:ETsizeBound}
  Let $\phi$ be an \msoo-sentence using $q=q_v+q_S$ quantifiers. 
  There is a computable function $f$ such that for each node $v$ of the elimination tree $F$ of width $\width$, the size of the compact evaluation tree of $v$ is bounded by $f(q+\width)\cdot \log(n)$.
\end{lemma}

\begin{proof}
  We proceed by induction on $q$ levels of the compact evaluation tree.
  We already observed that the number of non-isomorphic configurations of leaves is bounded by $f'(q+\omega)$ for some computable function $f'$.
  Now, assume that such number is bounded by $b$ for $q'<q$ level compact evaluation subtree we derive that it is bounded for $q'+1$ level compact evaluation subtree by $2^b$.
  Indeed, the node $x$ can have $b$ non-isomorphic types of children, so the type of $x$ is determined by the subset of the isomorphism types present as its children.
  Hence, the final bound $f$ is the tower function of $q$ levels of the bound $f'$ imposed on the number of non-isomorphic leaves.

  The factor $\log n$ comes from the markers used for vertices in $\str(v)$ where $v$ is the root of $F$ while their number is bounded by $\width$ which was used to determine the number of different configurations in the leaves.
\end{proof}

Now, we get to the core of this technique.
We show how we can combine two configuration evaluation trees.

\begin{definition}[Tree product]\label{def:treeProduct}
  Let $\Phi(z_1, \ldots, z_\ell)$ be an \msoo-formula in a prenex normal form, where $z_1, \ldots, z_\ell$ are variables already assigned and $z_{\ell + 1}, \ldots, z_q$ are quantified variables.
Let $G$ be a graph and let $F$ be an elimination tree for $G$. Let $t \in V(F)$ and let $t_1, t_2$ be distinct children of $t$.
Let $Q_1, Q_2$ be evaluation trees for $\Phi(z_1, \ldots, z_\ell)$ and $G_{t_1}$ and $G_{t_2}$ respectively.
Let $u,w$ denote the roots of $Q_1, Q_2$ respectively.
  The \emph{product} of $Q_1, Q_2$ is an evaluation tree $Q \coloneq Q_1 \otimes Q_2$ defined recursively.
  \begin{itemize}
    \item If the height of both $Q_1, Q_2$ is greater than $0$ we split the definition into two cases.
    \begin{itemize}
      \item \textbf{Case: Variable $z_{\ell + 1}$ is a set variable.} 
        Let be $Q_1^1, \ldots, Q_1^n$ children of $Q_1$ and $Q_2^1, \ldots, Q_2^m$ be children of $Q_2$.  
        Then $Q$ is a tree with children $Q_1^i \otimes Q_2^j$ for $i \in [n], j \in [m]$.  
        (Observe that $Q_1^i$ and $Q_2^j$ are evaluation trees for $\Phi(z_1, \ldots, z_\ell, z_{\ell + 1})$ of height one smaller that $Q_1, Q_2$, so their product is well-defined.)
      \item \textbf{Case: Variable $z_{\ell + 1}$ is an individual variable.}
        Let be $Q_1^0,Q_1^1, \ldots, Q_1^n$ children of $Q_1$ and  let $Q_2^0, Q_2^1, \ldots, Q_2^m$ be children of $Q_2$, where $Q_1^0$ and $Q_2^0$ correspond to the edge labeled by $ext$ in $Q_1,Q_2$, respectively.
        We split into two more subcases:
        \begin{itemize}
          \item For each $u\in\str(t)$, whenever there is $i\in n$ and $j\in[m]$ such that both $Q_1^i$ and $Q_2^j$ correspond to an edge marked by $u$ then we append $Q_1^i \otimes Q_2^j$ to set of children of $Q$ marking the corresponding edge by $u$.
            We remove the mark in case $u=t$.
            We remove both $Q_1^i$ and $Q_2^j$ from consideration for the second subcase.
          \item  For each $j\in[m]$ not excluded in the previous case we append $Q_1^0 \otimes Q_2^j$ to the set of children of $Q$.
          For each $i\in[n]$ not excluded in the previous case we append $Q_1^i \otimes Q_2^0$ to the set of children of $Q$.
          We append $Q_1^0 \otimes Q_2^0$ labeled by $ext$ to set of children of $Q$
    \end{itemize}
    \end{itemize}
    \item If the height of both $Q_1, Q_2$ is $0$ (that means $\ell = q$) and both $Q_1$ and $Q_2$ are represented by a configuration with underlying graphs $G_1$, $G_2$, respectively, then $Q$ is a tree of height $0$ with the following configuration:
      \begin{itemize}
        \item The configuration graph identifies vertices marked by the same vertex in $\str(t)$.
          We denote the set of such vertices as $X$.
          Then it performs the disjoint union of the remaining vertices in $V(G_1)$ and $V(G_2)$.
          All edges and non-edges of $G_1$ and $G_2$ are kept.
          We mark the relations between $V(G_1)\setminus X$ and $V(G_2)\setminus$ as non-edges.
          Therefore, we maintain the property that all relations are known except between vertices in $\str(t)$.
          Knowing $N_G(t)$ we specify all the relations between $t$ and $\str(t)$ to be either edge or non-edge.
          We denote the resulting graph as $G$.
        \item The new assignment of values from $V(G)\cup \{ext\}$ to the individual variables is a composition of the previous assignments where those are conflicting (due to the recursive construction) only when one of the original assignments assigned $ext$ and the other assigned non-$ext$ value. 
          In such a case, we keep only the non-$ext$ value.
        \item  The new set assignment will be a composition of the previous assignments.
      \end{itemize}
  \end{itemize}
\end{definition}

Observe that the two properties of configuration will be maintained in the base case of the recursion.
Indeed, we unified vertices in $\str(t)$, so the first property is still maintained in graph $G$.
The second property also follows easily from the construction of $G$.
It is easy to verify that the tree product preserves the consistency requirements in the definition of a configuration evaluation tree.
Therefore the product of two configuration evaluation trees is a configuration evaluation tree.

Now, we are ready to prove \cref{thm:mcMSO1onEF}.

\begin{proof}[Proof of \cref{thm:mcMSO1onEF}]
The proof is done by induction bottom-up on the elimination tree $F$.
When $t$ is a leaf of $F$ we construct a partial evaluation tree on vertices $\str(t)$ and then we reduce the corresponding configuration tree.
Therefore, we obtain a compact evaluation tree $Q_t$.
This easily satisfies Properties \hyperref[def:1]{1)} (due to \cref{lem:ETsizeBound}) and \hyperref[def:2]{2)}

When $t$ is an inner node of $F$ with children $t_1,t_2,\ldots$ the situation is a bit more complicated.
We first create an auxiliary evaluation tree containing vertices in $\str(t)\setminus\bigcup_{i\in \{1,2,\ldots\}} \str(t_i)$.
We denote this tree as $A_t$.
Then we compute $Q_{t_1} \otimes Q_{t_2}$ and reduce the result to obtain $Q'_{t_2}$.
In general, if $Q_{t_i}$ is the smallest $i$ such that $Q_{t_i}$ was not merged yet then we perform $Q'_{t_{i-1}} \otimes Q_{t_i}$ and then reduce to create $Q'_{t_i}$.
As a final step, once $Q_{t_i}$ is merged for all $i$, we take largest $j$ such that $Q'_{t_j}$ exists and perform $Q'_{t_j}\otimes A_t$.
Then we reduce it to obtain $Q'_t$.

Observe that as we reduced after each step, Property \hyperref[def:1]{1)} is satisfied by \cref{lem:ETsizeBound}.
It follows from the construction and it is easy to verify that $Q_t$ satisfies Property \hyperref[def:2]{2)}.
To show Property \hyperref[def:0]{0)}, we utilize the following claim. 

\begin{claim}[Combining lemma; analog of {\cite[Lemma 5.20]{DM24}}] \label{cl:combining-lemma}
 Let  $t_1, t_2$ be distinct nodes in $F$ with the same parent, $Q_1 \sim Q_1'$ 
 be configuration evaluation trees for $\Phi$ and $t_1$, $Q_2 \sim Q_2' $ be configuration evaluation trees for $\Phi$ and $t_2$.
  Then $Q_1 \otimes Q_2 \sim Q_1' \otimes Q_2'$.
\end{claim}

\begin{proof}\renewcommand\qedsymbol{$\diamondsuit$}
The proof of combining lemma follows almost identically the proof of \cite[Lemma 5.20]{DM24}.
\end{proof}

In other words, \cref{cl:combining-lemma} states that the product of two compact evaluation trees is isomorphic to the product of the respective partial evaluation trees.
Let $Q_1,Q_2$ be two partial evaluation trees (with specified vertices in $\str(v)$) such that $Q_1,Q_2$ describe disjoint graphs (except configurations in leaves of both $Q_1$ and $Q_2$ can be marked by $\str(v)$ for some $v\in G$).
Then observe that the $Q_1\otimes Q_2$ results in the partial evaluation tree (with specified vertices marked by $\str(v)$) of the disjoint union of the corresponding graphs.
Also, observe that the case when the adjacencies between $v$ and $\str(v)$ are determined and subsequently the markings by $v$ are removed also preserves the definition of partial evaluation trees.
Hence, we conclude that we can build a partial evaluation tree of $G$ with induction bottom-up on the elimination tree using the tree product $\otimes$.
However, the partial evaluation tree in the root is the full evaluation tree (we can omit any branches containing external vertices).
Therefore, it can be evaluated to determine the truthfulness of the formula $\phi$.
Using \cref{cl:combining-lemma} we can replace the partial evaluation tree at each step of the elimination tree with its isomorphic compact evaluation tree and perform the tree product $\otimes$ on the compact evaluation trees instead.
By induction on the tree decomposition bottom up, we obtain a compact evaluation tree of the root which is isomorphic by \cref{cl:combining-lemma} to the full evaluation tree of the root.
But again, a compact evaluation tree of the root can be evaluated as we can ignore external vertices and there are no vertices in $\str(v)\setminus\{v\}$ where $v$ is a root of elimination tree $F$.
Therefore we conclude the proof as we have shown Property \hyperref[def:0]{0)} holds.
\end{proof}

\bibliographystyle{alphaurl}
\bibliography{biblioLOCAL, biblio}

\newcommand{\etalchar}[1]{$^{#1}$}
\begin{thebibliography}{FMM{\etalchar{+}}23}

\bibitem[BC25]{pathwidth}
Dan~Alden Baterisna and Yi-Jun Chang.
\newblock Optimal local certification on graphs of bounded pathwidth, 2025.
\newblock \href {https://arxiv.org/abs/2502.00676} {\path{arXiv:2502.00676}}.

\bibitem[BCH{\etalchar{+}}24]{squirrels-tree-minors-2024}
Pablo Blanco, Linda Cook, Meike Hatzel, Claire Hilaire, Freddie Illingworth,
  and Rose McCarty.
\newblock On tree decompositions whose trees are minors.
\newblock {\em Journal of Graph Theory}, 106(2):296--306, 2024.
\newblock \href {https://doi.org/10.1002/jgt.23083}
  {\path{doi:10.1002/jgt.23083}}.

\bibitem[BFP24]{bousquet-small-h}
Nicolas Bousquet, Laurent Feuilloley, and Théo Pierron.
\newblock Local certification of graph decompositions and applications to
  minor-free classes.
\newblock {\em Journal of Parallel and Distributed Computing}, 193:104954,
  2024.
\newblock \href {https://doi.org/10.1016/j.jpdc.2024.104954}
  {\path{doi:10.1016/j.jpdc.2024.104954}}.

\bibitem[BGHK95]{BodlaenderGHK95}
Hans~L. Bodlaender, John~R. Gilbert, Hj{\'{a}}lmtyr Hafsteinsson, and Ton
  Kloks.
\newblock Approximating treewidth, pathwidth, frontsize, and shortest
  elimination tree.
\newblock {\em J. Algorithms}, 18(2):238--255, 1995.
\newblock \href {https://doi.org/10.1006/JAGM.1995.1009}
  {\path{doi:10.1006/JAGM.1995.1009}}.

\bibitem[BL95]{bienstock-algorithm-minors}
Daniel Bienstock and Michael~A. Langston.
\newblock Chapter 8 algorithmic implications of the graph minor theorem.
\newblock In {\em Network Models}, volume~7 of {\em Handbooks in Operations
  Research and Management Science}, pages 481--502. Elsevier, 1995.
\newblock \href {https://doi.org/10.1016/S0927-0507(05)80125-2}
  {\path{doi:10.1016/S0927-0507(05)80125-2}}.

\bibitem[Bod96]{Bodlaender96}
Hans~L. Bodlaender.
\newblock A linear-time algorithm for finding tree-decompositions of small
  treewidth.
\newblock {\em {SIAM} J. Comput.}, 25(6):1305--1317, 1996.
\newblock \href {https://doi.org/10.1137/S0097539793251219}
  {\path{doi:10.1137/S0097539793251219}}.

\bibitem[BP16]{BojanczykP2016}
Miko{\l}aj Boja\'{n}czyk and Micha{\l} Pilipczuk.
\newblock Definability equals recognizability for graphs of bounded treewidth.
\newblock In {\em Proceedings of the 31st {A}nnual {ACM}-{IEEE} {S}ymposium on
  {L}ogic in {C}omputer {S}cience ({LICS} 2016)}, pages 407--416. ACM, New
  York, 2016.
\newblock \href {https://doi.org/10.1145/2933575.2934508}
  {\path{doi:10.1145/2933575.2934508}}.

\bibitem[BP22]{BojanczykP2017}
Miko{\l}aj Boja{\'n}czyk and Micha{\l} Pilipczuk.
\newblock Optimizing tree decompositions in {MSO}.
\newblock {\em Logical Methods in Computer Science}, Volume 18, Issue 1,
  February 2022.
\newblock \href {https://doi.org/10.46298/lmcs-18(1:26)2022}
  {\path{doi:10.46298/lmcs-18(1:26)2022}}.

\bibitem[CE12]{CourcelleE2012}
Bruno Courcelle and Joost Engelfriet.
\newblock {\em Graph structure and monadic second-order logic}, volume 138 of
  {\em Encyclopedia of Mathematics and its Applications}.
\newblock Cambridge University Press, Cambridge, 2012.
\newblock A language-theoretic approach, With a foreword by Maurice Nivat.
\newblock \href {https://doi.org/10.1017/CBO9780511977619}
  {\path{doi:10.1017/CBO9780511977619}}.

\bibitem[CHPP20]{CENSORHILLEL-approximatePLS}
Keren Censor-Hillel, Ami Paz, and Mor Perry.
\newblock Approximate proof-labeling schemes.
\newblock {\em Theoretical Computer Science}, 811:112--124, 2020.
\newblock Special issue on Structural Information and Communication Complexit.
\newblock \href {https://doi.org/10.1016/j.tcs.2018.08.020}
  {\path{doi:10.1016/j.tcs.2018.08.020}}.

\bibitem[Cou90a]{Courcelle90twDP}
Bruno Courcelle.
\newblock {\em Graph Rewriting: An Algebraic and Logic Approach}, pages
  193--242.
\newblock Elsevier, 1990.
\newblock \href {https://doi.org/10.1016/b978-0-444-88074-1.50010-x}
  {\path{doi:10.1016/b978-0-444-88074-1.50010-x}}.

\bibitem[Cou90b]{Courcelle90}
Bruno Courcelle.
\newblock The monadic second-order logic of graphs. {I}. {R}ecognizable sets of
  finite graphs.
\newblock {\em Inform. and Comput.}, 85(1):12--75, 1990.
\newblock \href {https://doi.org/10.1016/0890-5401(90)90043-H}
  {\path{doi:10.1016/0890-5401(90)90043-H}}.

\bibitem[DHJ{\etalchar{+}}21]{product-structure}
Zdeněk Dvořák, Tony Huynh, Gwenael Joret, Chun-Hung Liu, and David~R. Wood.
\newblock Notes on graph product structure theory, 2021.
\newblock \href {https://doi.org/10.1007/978-3-030-62497-2_32}
  {\path{doi:10.1007/978-3-030-62497-2_32}}.

\bibitem[Die05]{Diestel2005}
Reinhard Diestel.
\newblock {\em Graph theory}, volume 173 of {\em Graduate Texts in
  Mathematics}.
\newblock Springer-Verlag, Berlin, third edition, 2005.
\newblock \href {https://doi.org/10.1007/978-3-662-70107-2}
  {\path{doi:10.1007/978-3-662-70107-2}}.

\bibitem[DM24]{DM24}
Karolina Drabik and Tom{\'a}{\v s} Masa{\v r}{\'\i}k.
\newblock Finding diverse solutions parameterized by cliquewidth, 2024.
\newblock \href {https://arxiv.org/abs/2405.20931} {\path{arXiv:2405.20931}}.

\bibitem[EL22]{louis-graphsonsurfaces}
Louis Esperet and Benjamin Lévêque.
\newblock Local certification of graphs on surfaces.
\newblock {\em Theoretical Computer Science}, 909:68--75, 2022.
\newblock \href {https://doi.org/10.1016/j.tcs.2022.01.023}
  {\path{doi:10.1016/j.tcs.2022.01.023}}.

\bibitem[Epp25]{treewidth-intro}
David Eppstein.
\newblock What is treewidth?
\newblock {\em Notices of the American Mathematical Society}, 72(2), 2025.
\newblock \href {https://doi.org/10.1090/noti3043}
  {\path{doi:10.1090/noti3043}}.

\bibitem[Esp21]{esperet-minor-local-talk}
Louis Esperet.
\newblock Local certification of graphs.
\newblock Talk at \emph{New Perspectives in Colouring and Structure} Hosted
  Online by Banff International Research Station, 2021.
\newblock URL:
  \url{https://www.birs.ca/events/2021/5-day-workshops/21w5513/schedule}.

\bibitem[FBP22]{treedepth}
Laurent Feuilloley, Nicolas Bousquet, and Th\'{e}o Pierron.
\newblock What can be certified compactly? {C}ompact local certification of
  {MSO} properties in tree-like graphs.
\newblock In {\em Proceedings of the 2022 ACM Symposium on Principles of
  Distributed Computing}, PODC'22, pages 131--140, New York, NY, USA, 2022.
  Association for Computing Machinery.
\newblock \href {https://doi.org/10.1145/3519270.3538416}
  {\path{doi:10.1145/3519270.3538416}}.

\bibitem[Feu21]{feuilloley2021introduction}
Laurent Feuilloley.
\newblock Introduction to local certification.
\newblock {\em Discrete Mathematics \& Theoretical Computer Science}, vol. 23,
  no. 3, 2021.
\newblock \href {https://doi.org/10.46298/dmtcs.6280}
  {\path{doi:10.46298/dmtcs.6280}}.

\bibitem[FFM{\etalchar{+}}20]{feuilloley2020planar}
Laurent Feuilloley, Pierre Fraigniaud, Pedro Montealegre, Ivan Rapaport,
  {\'E}ric R{\'e}mila, and Ioan Todinca.
\newblock Compact distributed certification of planar graphs.
\newblock In {\em Proceedings of the 39th Symposium on Principles of
  Distributed Computing}, pages 319--328, 2020.
\newblock \href {https://doi.org/10.1145/3382734.3404505}
  {\path{doi:10.1145/3382734.3404505}}.

\bibitem[FFM{\etalchar{+}}23]{feuilloley2023boundedgenus}
Laurent Feuilloley, Pierre Fraigniaud, Pedro Montealegre, Ivan Rapaport, Eric
  R{\'e}mila, and Ioan Todinca.
\newblock Local certification of graphs with bounded genus.
\newblock {\em Discrete Applied Mathematics}, 325:9--36, 2023.
\newblock \href {https://doi.org/10.1016/j.dam.2022.10.004}
  {\path{doi:10.1016/j.dam.2022.10.004}}.

\bibitem[Fia24]{Fiala}
Ji{\v r}{\'\i} Fiala.
\newblock Graph minors, decompositions and algorithms, April 2024.
\newblock Lecture notes, published in IUUK-CE-ITI series in 2003 number
  2003-132.
\newblock URL: \url{https://kam.mff.cuni.cz/~fiala/ASTG/tw.pdf}.

\bibitem[FMM{\etalchar{+}}23]{p4-cw}
Pierre Fraigniaud, Fr{\'e}d{\'e}ric Mazoit, Pedro Montealegre, Ivan Rapaport,
  and Ioan Todinca.
\newblock Distributed certification for classes of dense graphs, 2023.
\newblock \href {https://arxiv.org/abs/2307.14292} {\path{arXiv:2307.14292}}.

\bibitem[FMRT24]{logn-squared-treewidth-2024}
Pierre Fraigniaud, Pedro Montealegre, Ivan Rapaport, and Ioan Todinca.
\newblock A meta-theorem for distributed certification.
\newblock {\em Algorithmica}, 86(2):585--612, 2024.
\newblock \href {https://doi.org/10.1007/s00453-023-01185-1}
  {\path{doi:10.1007/s00453-023-01185-1}}.

\bibitem[GS16]{GoosSuomela16-locally-checkable-proofs}
Mika G{\"o}{\"o}s and Jukka Suomela.
\newblock Locally checkable proofs in distributed computing.
\newblock {\em Theory of Computing}, 12(19):1--33, 2016.
\newblock \href {https://doi.org/10.4086/toc.2016.v012a019}
  {\path{doi:10.4086/toc.2016.v012a019}}.

\bibitem[KK06]{KormanKutten-spanningtrees}
Amos Korman and Shay Kutten.
\newblock Distributed verification of minimum spanning trees.
\newblock In {\em Proceedings of the Twenty-Fifth Annual ACM Symposium on
  Principles of Distributed Computing}, PODC '06, page 26–34, New York, NY,
  USA, 2006. Association for Computing Machinery.
\newblock \href {https://doi.org/10.1145/1146381.1146389}
  {\path{doi:10.1145/1146381.1146389}}.

\bibitem[KKP10]{PLS}
Amos Korman, Shay Kutten, and David Peleg.
\newblock Proof labeling schemes.
\newblock {\em Distributed Computing}, 22(4):215--233, 2010.
\newblock \href {https://doi.org/10.1145/1073814.1073817}
  {\path{doi:10.1145/1073814.1073817}}.

\bibitem[Lag98]{Lagergren98}
Jens Lagergren.
\newblock Upper bounds on the size of obstructions and intertwines.
\newblock {\em J. Comb. Theory {B}}, 73(1):7--40, 1998.
\newblock \href {https://doi.org/10.1006/JCTB.1997.1788}
  {\path{doi:10.1006/JCTB.1997.1788}}.

\bibitem[Lin87]{linial-OG-LOCAL-paper-1987}
Nathan Linial.
\newblock Distributive graph algorithms global solutions from local data.
\newblock In {\em Proceedings of the 28th Annual Symposium on Foundations of
  Computer Science}, SFCS '87, pages 331--335, USA, 1987. IEEE Computer
  Society.
\newblock \href {https://doi.org/10.1109/SFCS.1987.20}
  {\path{doi:10.1109/SFCS.1987.20}}.

\bibitem[Lin92]{linial-1992-LOCAL-OG-journal-paper}
Nathan Linial.
\newblock Locality in distributed graph algorithms.
\newblock {\em SIAM Journal on Computing}, 21(1):193--201, 1992.
\newblock \href {https://doi.org/10.1137/0221015} {\path{doi:10.1137/0221015}}.

\bibitem[Lov06]{lovasz-graphminor-survey}
L{\'a}szl{\'o} Lov{\'a}sz.
\newblock Graph minor theory.
\newblock {\em Bulletin of the American Mathematical Society}, 43(1):75--86,
  2006.
\newblock \href {https://doi.org/10.1090/S0273-0979-05-01088-8}
  {\path{doi:10.1090/S0273-0979-05-01088-8}}.

\bibitem[NdM12]{NOdM2012}
Jaroslav Ne{\v s}et{\v r}il and Patrice~Ossona de~Mendez.
\newblock {\em Sparsity - Graphs, Structures, and Algorithms}, volume~28 of
  {\em Algorithms and combinatorics}.
\newblock Springer, 2012.
\newblock \href {https://doi.org/10.1007/978-3-642-27875-4}
  {\path{doi:10.1007/978-3-642-27875-4}}.

\bibitem[RS83]{graphminors1-excludingaforest}
Neil Robertson and P.D. Seymour.
\newblock Graph minors. {I.} excluding a forest.
\newblock {\em Journal of Combinatorial Theory, Series B}, 35(1):39--61, 1983.
\newblock \href {https://doi.org/10.1016/0095-8956(83)90079-5}
  {\path{doi:10.1016/0095-8956(83)90079-5}}.

\bibitem[RS86]{graphminorsV}
Neil Robertson and P.D Seymour.
\newblock Graph minors. {V.} excluding a planar graph.
\newblock {\em Journal of Combinatorial Theory, Series B}, 41(1):92--114, 1986.
\newblock \href {https://doi.org/10.1016/0095-8956(86)90030-4}
  {\path{doi:10.1016/0095-8956(86)90030-4}}.

\bibitem[RS04]{graphminors-wqo}
Neil Robertson and P.D. Seymour.
\newblock Graph minors. {XX.} {W}agner's conjecture.
\newblock {\em Journal of Combinatorial Theory, Series B}, 92(2):325--357,
  2004.
\newblock Special Issue Dedicated to Professor W.T. Tutte.
\newblock \href {https://doi.org/10.1016/j.jctb.2004.08.001}
  {\path{doi:10.1016/j.jctb.2004.08.001}}.

\bibitem[Sch82]{Schreiber82}
Robert Schreiber.
\newblock A new implementation of sparse gaussian elimination.
\newblock {\em {ACM} Trans. Math. Softw.}, 8(3):256--276, 1982.
\newblock \href {https://doi.org/10.1145/356004.356006}
  {\path{doi:10.1145/356004.356006}}.

\end{thebibliography}

\end{document}